\documentclass[runningheads,a4paper]{llncs}

\usepackage[english]{babel}
\usepackage[ansinew]{inputenc}
\usepackage{amsmath}
\usepackage{amsfonts}
\usepackage{amssymb}
\usepackage{graphicx}
\usepackage{fancybox}
\usepackage{calc}
\usepackage{tabularx}
\usepackage{tikz}
\usepackage{array}
\usepackage{lmodern}
\usepackage{mathabx}
\usepackage{fancybox}
\usepackage{fullpage}

\usepackage[margin=1.6in]{geometry}
\usepackage{tkz-berge}
\usetikzlibrary{fit,shapes}

\newcommand{\Mod}{\ \mathrm{mod}\ }
\newcommand{\xdownarrow}[1]{%
  {\left\downarrow\vbox to #1{}\right.\kern-\nulldelimiterspace}}

\newtheorem{Thm}{Theorem}
\newtheorem{Lem}[Thm]{Lemma}

\newtheorem{Def}{Definition}

\title{Investcoin:\\ A System for Privacy-Preserving Investments
\thanks{The research was supported by the DFG Research Training Group GRK $1817/1$}
}
\author{
Filipp Valovich\\ Email: filipp.valovich@rub.de
}
\institute{
Horst G\"{o}rtz Institute for IT Security\\ Faculty of Mathematics\\ Ruhr-Universit\"{a}t Bochum, Universit\"{a}tsstra{\ss}e 150, 44801 Bochum, Germany
}
\date{}

\allowdisplaybreaks

\begin{document}

\maketitle

\pagenumbering{gobble}
\pagestyle{empty}

\vspace{-1cm}
\begin{abstract}
This work presents a new framework for Privacy-Preserving Investment systems in a distributed model. In this model, independent investors can transfer funds to independent projects, in the same way as it works on crowdfunding platforms. The framework protects the investors' single payments from being detected (by any other party), only the sums of each investor's payments are revealed (e.g to the system). Likewise, the projects' single incoming payments are concealed and only the final sums of the incoming payments for every project are revealed. In this way, no other party than the investor (not even the system administration) can detect how much she paid to any single project. Though it is still possible to confidentially exchange any part of an investment between any pair of investors, such that market liquidity is unaffected by the system. On top, our framework allows a privacy-preserving return of a multiple of all the held investments (e.g. interest payments or dividends) to the indivdual investors while still revealing nothing else than the sum of all returns for every investor. We provide reasonable security guarantees for this framework that are based on common notions from the Secure Multi-Party Computation (SMPC) literature. As an instantiation for this framework we present Investcoin. This is a proper combination of three cryptographic protocols, namely a Private Stream Aggregation scheme, a Commitment scheme and a Range test and it is usable in connection with any existing currency. The security of the three protocols is based on the DDH assumption. Thus, by a composition theorem from the SMPC literature, the security of the resulting Investcoin protocol is also based on the DDH assumption. Furthermore, we provide a simple decentralised key generation protocol for Investcoin that supports dynamic join, dynamic leave and fault-tolarance of investors and moreover achieves some security guarantees against malicious investors.
\end{abstract}

\section{Introduction}

%user privacy in general is important - so why not protecting financial data, requirement to trust in systems is an obstacle for trade secrecy, transparency/privacy trade-off in organised markets, privacy for investors necessary, privacy issues influence investment decisions \cite{62} - individuals will paricipate more in investments on exchange or crowfunding platforms for instance if their privacy is protected, on the other hand reporting is more reliable when investors are sure about their privacy \cite{63,64,65}, sometimes possibility of circumvention of certain regulatories such as financial sanctions in repressive countries is desired, our objectives: solve privacy issues for investments by concealing particular investment decisions but offering transparency of aggregated investment decisions while maintaining market liquidity, therefore we define DIE schemes; PPI systems; Investcoin as particular PPI system - combination of PSA and Commitment schemes; give examples of Investcoin instantiations, goals achieved with Investcoin: only aggregates known - not single investments; market liquidity; by a simple secret sharing key protocol investors can join or leave or fail and investors cannot cheat maliciously; differential privacy possible but not recommended

The promise of performance benefit by using technologies like online-outsourcing and cloud-computing goes along with the loss of control over individual data. Therefore the public awareness of data protection increases. We use encryption and privacy technologies to protect our electronic messages, our consumer behaviour or patient records. In this work, we put the following question up for discussion: why is there only minor attention paid to the protection of sensitive \textit{financial data} in the public? Indeed the requirement to trust in financial institutes may be an obstacle for the trade secrecy of companies. On the one hand, transactions on organised markets are registered by electronic systems, audited and eventually get under the control of the system administration (e.g. it can refuse a transaction). In some cases this is desired: e.g. it should be possible to detect a company financing criminal activities. On the other hand, we would like to protect the trade secrecy of the companies. In this sense, there is a transparency/confidentiality trade-off in organised markets, such as exchange or to some extent also crowdfunding platforms.\\
In this work we address the problem of providing adequate privacy guarantees to investors. As observed by Nofer \cite{62}, although there is no observable significant effect concerning "the impact of privacy violations on the investment amount, (...) one has to remember that trust influences behavior (...) and privacy issues influences trust (...) and therefore an indirect influence still exists". Conversely, this means that individuals would participate more in investments if their privacy is protcted. As an effect, the reporting of investors concerning wins and losses (and therefore risks) becomes more reliable \cite{63,64,65}. As further motivation of our work, the possibility to circumvent certain regulatories, such as financial sanctions by order of "repressive" countries, may be desired. Investors may look for a way to invest in sanctioned companies without being traced by their home country.\\
Consequently, the objective of this work is to solve privacy issues by concealing particular investment decisions but offering transparency of "aggregated" investment decisions. In this regard we introduce a Distributed Investment Encryption (DIE) scheme for the aggregation of the investments of a number of different investors funding different projects on an electronic platform. A DIE scheme maintains market liquidity, i.e. the scheme does not affect the possibility to trade assets among investors. Its cryptographic security is set through the definition of the Privacy-Preserving Investment ($\text{\upshape\sffamily PPI}$) system. Informally, a $\text{\upshape\sffamily PPI}$ system conceals the single payments of investors from being traced but reveals to the system only the aggregates of the investors' payments. Similarly, the projects' single incoming payments are concealed and only the final sums of the incoming payments for every project are revealed. Moreover, a $\text{\upshape\sffamily PPI}$ system conceals the single returns (representing interest payments, coupons or dividends) from every single project to every single investor but reveals (to the system administration) the aggregated return of every single investor. Therefore, up to a certain extent, a $\text{\upshape\sffamily PPI}$ system simultaneously maintains transparency (e.g. taxes on the final return of every investor can be raised) and trade secrecy.\\
As a particular $\text{\upshape\sffamily PPI}$ system we present Investcoin, a combination of three cryptographic protocols: a Private Stream Aggregation (PSA) scheme, first introduced by Shi et al. \cite{2}, a (homomorphic) Commitment scheme and a Range test for commited values, all secure under the Decisional Diffie-Hellman assumption. Informally, the PSA scheme is used for the secure aggregation of funds for every particular project and the homomorphic Commitment scheme is used for the secure aggregation of all investments and returns of every particular investor. The Range test ensures that investments are not negative. We provide a simple secret sharing key generation protocol for Investcoin, that allows investors to dynamically join, leave or fail during the protocol execution and prevents investors from some malicious cheating.\\

%Commitments, Privacy-preserving data aggregation, Cryptocurrencies, Security/Privacy and finance: market regulation through aggregated privacy-preserving risk/information reporting while preserving trade secrecy \cite{67}; protocols, anonymity, differential privacy for balancing transparency and confidentiality \cite{65,66}

\noindent\textbf{Related work.} The notion of Commitment schemes (first in \cite{72,73}) is well-established in the literature. The notion of PSA was introduced in \cite{2}. A PSA scheme is a cryptographic protocol which enables a number of users to individually and securely send encrypted time-series data to an untrusted aggregator requiring each user to send exactly one message per time-step. The aggregator is able to decrypt the aggregate of all data per time-step, but cannot retrieve any further information about the individual data. In \cite{2} a security definition for PSA and a secure instantiation were provided. Joye and Libert \cite{38} provided a scheme with a tighter security reduction and Benhamouda et al. \cite{38seq} generalised the scheme in \cite{38}. By lowering the security requirements established in \cite{2}, Valovich and Ald{\`{a}} \cite{40} provided general conditions for the existence of secure PSA schemes, based on key-homomorphic weak PRFs.\\ 
Investcoin is not a classical cryptocurrency. It can be thought of as a cryptographic layer on top of any currency used for investments, similar to what Zerocoin is intended to be for Bitcoin. Bitcoin is the first cryptocurrency, introduced by Satoshi Nakamoto \cite{76}. Currently it is the cryptocurrency with the largest market capitalisation. Bitcoin is as a peer-to-peer payment system where transactions are executed directly between users without interaction of any intermediate party. The transactions are verified by the users of the network and publicly recorded in a blockchain, a distributed database. Zerocoin was proposed by Miers et al. \cite{77} as an extension for Bitcoin (or any other cryptocurrency) providing cryptographic anonymity to recorded transactions in the blockchain (Bitcoin itself provides only pseudonymity). This is achieved by the use of a seperate mixing procedure based on Commitment schemes. Therefore particular transactions cannot be publicly traced back to particular Bitcoin adresses anymore. This is also the main principle of a $\text{\upshape\sffamily PPI}$ system: no investment in a particular project can be traced back to a particular investor. In this regard Investcoin has similarities with Zerocoin.\\
Methods for market regulation through aggregated privacy-preserving risk reporting were studied by Abbe et al. \cite{67}. They constructed protocols allowing a number of users to securely compute aggregated risk measures based on summations and inner products. Flood et al. \cite{66} considered balancing transparency and confidentiality for financial regulation by investigating cryptographic tools for statistical data privacy.

\section{Preliminaries}\label{invprelsec}

In this section we provide the description of our investment model, the basic protocols underlying Investcoin and their corresponding security definitions.

%\begin{Ntn} Let $X$ be a set. If $X$ is finite, we denote by $_R X$ the uniform distribution on $X$. Let $\chi$ be a distribution (on $X$). We denote by $x\leftarrow\chi$ (or sometimes $x\leftarrow\chi(X)$) the sampling of $x$ (from $X$) according to $\chi$. If $X$ has only two elements, we write $x\leftarrow_R X$ instead of $x\leftarrow_R X$. If $A\leftarrow\chi^{a\times b}$ (or $A\leftarrow\chi(X^{a\times b})$) then $A$ is an $a\times b$-matrix constructed by picking every entry independently (from $X$) according to the distribution $\chi$.
%\end{Ntn}

%\begin{Ntn} Let $\kappa$ be a security parameter. If $\omega=\omega(\kappa)<1/\text{poly}(\kappa)$ for every polynomial $\text{poly}$ and all $\kappa>\kappa^\prime$ for some $\kappa^\prime\in\mathbb{N}$ then we say that $\omega$ is negligible in $\kappa$ and denote it by $\omega=\text{\upshape\sffamily neg}(\kappa)$. If $\omega$ is non-negligible in $\kappa$ (i.e. if $\omega> 1/\text{poly}(\kappa)$), then we write $\omega>\text{\upshape\sffamily neg}(\kappa)$.
%\end{Ntn}

\subsection{Model}\label{invmodelsec}
%bipartite graph, crowdfunding, exchange, investors, projects, rounds, payment, payback, fees, system as pool for distribution of funds and refunds, description of goals (security, correctness, communication complexity, market liquidity, dynamics, fault-tolerance, maliciousness)

As initial situation we consider a network consisting of $n$ investors and $\lambda$ projects to be funded by the investors. As an analogy from the real world one can think of a crowdfunding platform or an exchange system where projects or companies try to collect funds from various individual investors. Each investor $N_i$, $i=1,\ldots,n$, is willing to invest the amount $x_{i,j}\geq 0$ to the project $P_j$, $j=1,\ldots,\lambda$, thus the total amount invested by $N_i$ is $\sum_{j=1}^\lambda x_{i,j}$ and the total amount received by project $P_j$ is $\sum_{i=1}^n x_{i,j}$. Moreover, there exists an administration (which may be the administration of the crowdfunding platform). The investors and the project managements \textit{are not required to trust the administration}.\\
We consider a series of investment rounds. An investment round denotes the moment when the payments of all participating investors are registered by the administration of the system. From round to round the values $n$ and $\lambda$ may change, i.e. investors and projects may join or leave the network before any round.\\
After an investment round is over and the time comes to give a return to the investors (i.e. at maturity), the management of each project $P_j$ publishes some value $\alpha_j$ defining the return for each investor (i.e. an indicator of economic growth, interest yield, dividend yield or similar). The untrusted system administration (or simply system) serves as a pool for the distribution of investments to the projects and of returns to the investors: first, for all $i=1,\ldots,n$ it collects the total amount $\sum_{j=1}^\lambda x_{i,j}$ invested by investor $N_i$ and rearranges the union of the total amounts into the aggregated investment $\sum_{i=1}^n x_{i,j}$ for project $P_j$ for all $j=1,\ldots,\lambda$; at maturity date, for all $j=1,\ldots,\lambda$ it collects the total returns $\alpha_j\sum_{i=1}^n x_{i,j}$ of the projects and rearranges the union of the total returns into the returns $\sum_{j=1}^\lambda \alpha_j x_{i,j}$ of the investors.\\ 
While the investors do not have to trust each other nor the system (i.e. an investor doesn't want the others to know her financial data), we consider a honest-but-curious model where the untrusted system administration tries to compromise investors to build a coalition. This coalition tries to infer additional information about uncompromised investors, but under the constraint of honestly following the investment protocol. On the other hand, we allow investors that are not part of the coalition, to execute some malicious behaviour, that we will concretise later. Thereby we have the following objectives in each investment round:\\

\textit{Security}
\vspace{-0.2cm}
\begin{itemize}
\item Hiding: For all $i=1,\ldots,n$ the only financial data of investor $N_i$ (if uncompromised) known to the system is $C_i=\sum_{j=1}^\lambda x_{i,j}$ and $E_i=\sum_{j=1}^\lambda \alpha_j x_{i,j}$. For all $j=1,\ldots,\lambda$ the only financial data of project $P_j$ known to the system is $X_j=\sum_{i=1}^n x_{i,j}$ and $\alpha_j$. Particularly, the single investments of $N_i$ to $P_1,\ldots,P_\lambda$ should remain concealed.
\item Binding: Investors may not announce an incorrect investment, i.e. if $N_i$ has send $x_{i,j}$ to $P_j$, then $P_j$ should also receive $x_{i,j}$ from $N_i$.
\item For all $i,j$ it holds that $x_{i,j}\geq 0$, i.e. no investor can 'steal' money from a project.
\end{itemize}\par\medskip

\textit{Correctness}
\vspace{-0.2cm}
\begin{itemize}
\item For all $i$, if $C_i$ is the real aggregate of $N_i$'s investments, then $C_i=\sum_{j=1}^\lambda x_{i,j}$, the system knows $C_i$ and can charge the bank account of $N_i$ with amount $C_i$.
\item For all $j$, if $X_j$ is the real aggregate of $P_j$'s funds, then $X_j=\sum_{i=1}^n x_{i,j}$, the system knows $X_j$ and transfers the amount $X_j$ to the bank account of $P_j$.
\item For all $i$, if $E_i$ is the real aggregate of $N_i$'s returns, then $E_i=\sum_{j=1}^\lambda \alpha_j x_{i,j}$, the system knows $E_i$ and transfers the amount $E_i$ to the bank account of $N_i$.
\item If one of these conditions is violated (e.g. on the purpose of stealing money), then the injured party should be able to detect this fact and to prove it to the network latest after the end of the corresponding investment round.
\end{itemize}

Now we provide the building blocks for a scheme satisfying these objectives. 

\subsection{Private Stream Aggregation}\label{invpsasec}

In this section, we define Private Stream Aggregation (PSA) and provide a security definition. The notion of PSA was introduced by Shi et al. \cite{2}.

\subsubsection{The definition of Private Stream Aggregation.}

A PSA scheme is a protocol for safe distributed time-series data transfer which enables the receiver (here: the system administrator) to learn nothing else than the sums $\sum_{i=1}^n x_{i,j}$ for $j=1,2,\ldots$, where $x_{i,j}$ is the value of the $i$th participant in (time-)step $j$ and $n$ is the number of participants (here: investors). Such a scheme needs a key exchange protocol for all $n$ investors together with the administrator as a precomputation, and requires each investor to send exactly one message (namely the amount to spend for a particular project) in each step $j=1,2,\ldots$.

\begin{Def}[Private Stream Aggregation \cite{2}]
Let $\kappa$ be a security parameter, $\mathcal{D}$ a set and $n, \lambda\in\mathbb{N}$ with $n=\text{poly}(\kappa)$ and $\lambda=\text{poly}(\kappa)$. A \text{\upshape Private Stream Aggregation} (PSA) scheme $\Sigma=(\mbox{\upshape\sffamily Setup}, \mbox{\upshape\sffamily PSAEnc}, \mbox{\upshape\sffamily PSADec})$ is defined by three ppt algorithms:
\begin{description}
\item \textbf{\mbox{\upshape \sffamily Setup}}: $(\mbox{\upshape\sffamily pp},T,s_0,s_1,\ldots,s_n)\leftarrow \mbox{\upshape\sffamily Setup}(1^\kappa)$ with public parameters $\mbox{\upshape\sffamily pp}$, $T=\{t_1,\ldots,t_\lambda\}$ and secret keys $s_i$ for all $i=1,\ldots,n$.
\item \textbf{\mbox{\upshape \sffamily PSAEnc}}: For $t_j\in T$ and all $i=1,\ldots,n$: $c_{i,j}\leftarrow \mbox{\upshape\sffamily PSAEnc}_{s_i}(t_j,x_{i,j})\mbox{ for } x_{i,j}\in\mathcal{D}$.
\item \textbf{\mbox{\upshape \sffamily PSADec}}: Compute $\sum_{i=1}^n x'_{i,j}=\mbox{\upshape\sffamily PSADec}_{s_0}(t_j,c_{1,j},\ldots,c_{n,j})$ for $t_j\in T$ and ciphers $c_{1,j},\ldots,c_{n,j}$. For all $t_j\in T$ and $x_{1,j},\ldots,x_{n,j}\in\mathcal{D}$ the following holds:
\[\mbox{\upshape\sffamily PSADec}_{s_0}(t_j, \mbox{\upshape\sffamily PSAEnc}_{s_1}(t_j,x_{1,j}),\ldots,\mbox{\upshape\sffamily PSAEnc}_{s_n}(t_j,x_{n,j}))=\sum_{i=1}^n x_{i,j}.\]
\end{description}
\end{Def}

The system parameters $\mbox{\upshape\sffamily pp}$ are public and constant for all $t_j$ with the implicit understanding that they are used in $\Sigma$. Every investor encrypts her amounts $x_{i,j}$ with her own secret key $s_i$ and sends the ciphertext to the administrator. If the administrator receives the ciphertexts of \textit{all} investors for some $t_j$, it can compute the aggregate of the investors' data using the decryption key $s_0$.\\
While in \cite{2}, the $t_j\in T$ were considered to be time-steps within a time-series (e.g. for analysing time-series data of a smart meter), in our work the $t_j\in T$ are associated with projects $P_j$, $j=1,\ldots,\lambda$, to be funded in a particular investment round.

\subsubsection{Security of Private Stream Aggregation.}

Our model allows an attacker to compromise investors. It can obtain auxiliary information about the values of investors or their secret keys. Even then a secure PSA scheme should release no more information than the aggregates of the uncompromised investors' values.

\begin{Def}[Aggregator Obliviousness \cite{2}]\label{securitygame1}
Let $\kappa$ be a security parameter. Let $\mathcal{T}$ be a ppt adversary for a PSA scheme $\Sigma=$ $(\mbox{\upshape\sffamily Setup}, \mbox{\upshape\sffamily PSAEnc}, \mbox{\upshape\sffamily PSADec})$ and let $\mathcal{D}$ be a set. We define a security game between a challenger and the adversary $\mathcal{T}$.
\begin{description}
 \item\textbf{Setup.} The challenger runs the \text{\upshape\sffamily Setup} algorithm on input security parameter $\kappa$ and returns public parameters $\mbox{\upshape\sffamily pp}$, public encryption parameters $T$ with $|T|=\lambda=\text{poly}(\kappa)$ and secret keys $s_0,s_1,\ldots,s_n$. It sends $\kappa,\mbox{\upshape\sffamily pp}, T, s_0$ to $\mathcal{T}$.
\item\textbf{Queries.} $\mathcal{T}$ is allowed to query $(i,t_j,x_{i,j})$ with $i\in\{1,\ldots,n\}, t_j\in T, x_{i,j}\in\mathcal{D}$ and the challenger returns $c_{i,j}\leftarrow\mbox{\upshape\sffamily PSAEnc}_{s_i}(t_j,x_{i,j})$.
Moreover, $\mathcal{T}$ is allowed to make compromise queries $i\in\{1,\ldots,n\}$ and the challenger returns $s_i$.
\item\textbf{Challenge.} $\mathcal{T}$ chooses $U\subseteq\{1,\ldots,n\}$ such that no compromise query for $i\in U$ was made and sends $U$ to the challenger. $\mathcal{T}$ chooses $t_{j^*}\in T$ such that no encryption query with $t_{j^*}$ was made. (If there is no such $t_{j^*}$ then the challenger simply aborts.) $\mathcal{T}$ queries two different tuples $(x_{i,j^*}^{[0]})_{i\in U},(x_{i,j^*}^{[1]})_{i\in U}$ with
\[\sum_{i\in U} x_{i,j^*}^{[0]}=\sum_{i\in U} x_{i,j^*}^{[1]}.\]
The challenger flips a random bit $b\leftarrow_R\{0,1\}$. For all $i\in U$ the challenger returns $c_{i,j^*}\leftarrow\mbox{\upshape\sffamily PSAEnc}_{s_i}(t_{j^*},x_{i,j^*}^{[b]})$.
\item\textbf{Queries.} $\mathcal{T}$ is allowed to make the same type of queries as before restricted to encryption queries with $t_j\neq t_{j^*}$ and compromise queries for $i\notin U$.
\item\textbf{Guess.} $\mathcal{T}$ outputs a guess about $b$.
\end{description}
The adversary wins the game if it correctly guesses $b$. A PSA scheme achieves \text{\upshape Aggregator Obliviousness} ($\mbox{\upshape\sffamily AO}$) or is \text{\upshape secure} if no ppt adversary $\mathcal{T}$ has more than negligible advantage (with respect to the parameter $\kappa$) in winning the above game.
\end{Def}

Encryption queries are made only for $i\in U$, since knowing the secret key for all $i\notin U$ the adversary can encrypt a value autonomously. If encryption queries in time-step $t_{j^*}$ were allowed, then no deterministic scheme would be secure. The adversary $\mathcal{T}$ can determine the original data of all $i\notin U$, since it knows $(s_i)_{i\notin U}$. Then $\mathcal{T}$ can compute the sum $\sum_{i\in U} x_{i,j}=\mbox{\upshape\sffamily PSADec}_{s_0}(t_j,c_{1,j},\ldots,c_{n,j})-\sum_{i\notin U} x_{i,j}$ of the uncompromised investors' values. If there is an investor's cipher which $\mathcal{T}$ does not receive, then it cannot compute the sum for the corresponding $t_j$.\\
It is also possible to define $\mbox{\upshape\sffamily AO}$ in the non-adaptive model as in \cite{40,68}: here an adversary may not compromise investors adaptively, but has to specify the coalition $U$ of compromised investors \textit{before} making any query.

\subsubsection{Feasibility of \mbox{\upshape\sffamily AO}.}\label{rompsasec}

In the random oracle model we can achieve $\mbox{\upshape\sffamily AO}$ for some constructions \cite{2,38seq,40}. Because of its simplicity and efficient decryption, we use the PSA scheme proposed in \cite{40} and present it in Figure \ref{DDHEXM1}. It achieves $\mbox{\upshape\sffamily AO}$ based on the DDH assumption.

\begin{figure}[hpbt]
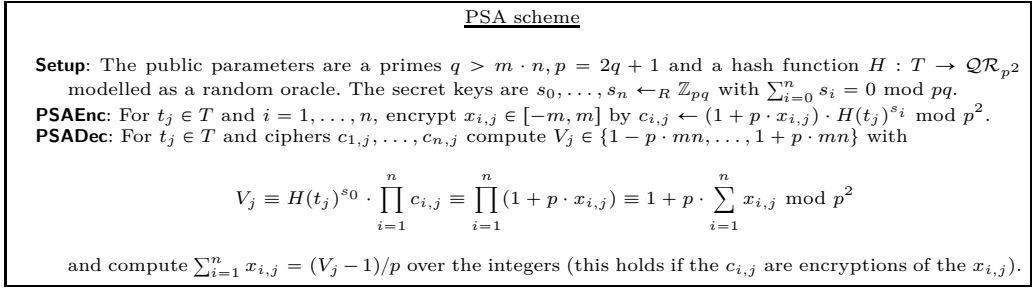

\fbox{
\begin{minipage}{13.15cm}
\begin{center}\scriptsize
\underline{PSA scheme}
\end{center}
\begin{description}\scriptsize
\item \textbf{\mbox{\upshape \sffamily Setup}}: The public parameters are a primes $q>m\cdot n,p=2q+1$ and a hash function $H: T\to\mathcal{QR}_{p^2}$ modelled as a random oracle. The secret keys are $s_0,\ldots,s_n\leftarrow_R\mathbb{Z}_{pq}$ with $\sum_{i=0}^n s_i=0$ mod $pq$.
\item \textbf{\mbox{\upshape \sffamily PSAEnc}}: For $t_j\in T$ and $i=1,\ldots,n$, encrypt $x_{i,j}\in [-m,m]$ by $c_{i,j}\leftarrow (1+p\cdot x_{i,j})\cdot H(t_j)^{s_i}$ mod $p^2$.
\item \textbf{\mbox{\upshape \sffamily PSADec}}: For $t_j\in T$ and ciphers $c_{1,j},\ldots,c_{n,j}$ compute $V_j\in\{1-p\cdot mn,\ldots,1+p\cdot mn\}$ with \[V_j\equiv H(t_j)^{s_0}\cdot\prod_{i=1}^n c_{i,j}\equiv\prod_{i=1}^n (1+p\cdot x_{i,j})\equiv 1+p\cdot\sum_{i=1}^n x_{i,j} \text{ mod } p^2\] and compute $\sum_{i=1}^n x_{i,j}=(V_j-1)/p$ over the integers (this holds if the $c_{i,j}$ are encryptions of the $x_{i,j}$).
\end{description}
\end{minipage}}
\caption{PSA scheme secure in the random oracle model.}\label{DDHEXM1}
\end{figure}

\normalsize
The scheme proposed in \cite{2} is similar to the one in Figure \ref{DDHEXM1}, but is inefficient, if the range of possible decryption outcomes in is super-polynomially large in the security parameter. The scheme in Figure \ref{DDHEXM1} also achieves the non-adaptive version of $\mbox{\upshape\sffamily AO}$ in the standard model.

\subsection{Commitment schemes}\label{invcomsec}

A Commitment scheme allows a party to publicly commit to a value such that the value cannot be changed after it has been committed to (binding) and the value itself stays hidden to other parties until the owner reveals it (hiding). For the basic definitions we refer to \cite{72} and \cite{73}. Here we just recall the Pedersen Commitment introduced in \cite{71} (Figure \ref{pedersenexm}), which is computationally binding under the dlog assumption and perfectly hiding. In Section \ref{invcnsec} we will combine the Pedersen Commitment with the PSA scheme from Figure \ref{DDHEXM1} for the construction of Investcoin and thereby consider the input data $x=x_{i,j}$ to the Commitment scheme as investment amounts from investor $N_i$ to project $P_j$.
An essential property for the construction of Investcoin is that the Pedersen Commitment contains a homomorphic commitment algorithm, i.e. $\mbox{\upshape\sffamily Com}_{pk}(x, r)*\mbox{\upshape\sffamily Com}_{pk}(x^\prime, r^\prime)=\mbox{\upshape\sffamily Com}_{pk}(x+x^\prime,r+r^\prime)$.

\begin{figure}[hpbt]
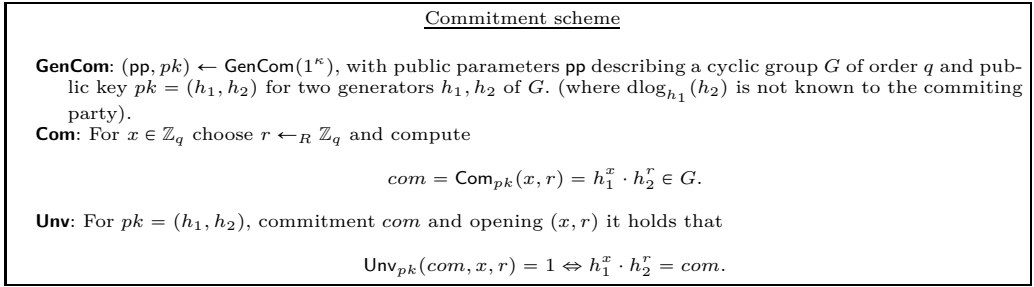

\fbox{
\begin{minipage}{13.15cm}
\begin{center}\scriptsize
\underline{Commitment scheme}
\end{center}
\begin{description}\scriptsize
\item \textbf{\mbox{\upshape \sffamily GenCom}}: $(\mbox{\upshape\sffamily pp},pk)\leftarrow \mbox{\upshape\sffamily GenCom}(1^\kappa)$, with public parameters $\mbox{\upshape\sffamily pp}$ describing a cyclic group $G$ of order $q$ and public key $pk=(h_1,h_2)$ for two generators $h_1,h_2$ of $G$. (where $\text{dlog}_{h_1}(h_2)$ is not known to the commiting party).
\item \textbf{\mbox{\upshape \sffamily Com}}: For $x\in\mathbb{Z}_q$ choose $r\leftarrow_R\mathbb{Z}_q$ and compute 
\[com=\mbox{\upshape\sffamily Com}_{pk}(x, r)=h_1^x\cdot h_2^r\in G.\]
\item \textbf{\mbox{\upshape \sffamily Unv}}: For $pk=(h_1,h_2)$, commitment $com$ and opening $(x,r)$ it holds that
\[\mbox{\upshape\sffamily Unv}_{pk}(com,x,r)=1\Leftrightarrow h_1^x\cdot h_2^r=com.\]
\end{description}
\end{minipage}}
\caption{The Pedersen Commitment.}\label{pedersenexm}
\end{figure}

\subsection{Range test}\label{rangesec}

To allow the honest verifier to verify in advance, that the (possibly malicious) prover commits to a an integer $x$ in a certain range, a Range test must be applied. Range tests were studied in \cite{56,57,58,59}. For Investcoin, an interactive proof procedure can be applied. It is a combination of the Pedersen Commitment to the binary representation of $x$ and the extended Schnorr proof of knowledge \cite{79} (Figure \ref{extschnorr}) applied to proving knowledge of one out of two secrets as described by Cramer \cite{80}. Its basic idea was described by Boudot \cite{56}.\\
Let $W(x)$ be a set of witnesses for $x$. We define the following relation.
\[R=\{(x,w)\,|\, x\in L, w\in W(x)\}.\]
An interactive proof procedure has the following properties. First, it is \textit{complete}, meaning that if $(x,w)\in R$ and the (possibly malicious) prover $P$ knows the witness $w$ for $x$, then it is able to convince the honest verifier $V$ (i.e. $V$ honestly follows the procedure). Second, it is \textit{sound}, meaning that if $(x,w)\notin R$, then no ppt algorithm $P$ is able to convince the honest $V$ except with negligible probability. For a proof of knowledge, instead of soundness we need \textit{special soundness}, i.e. there exists a ppt algorithm that extracts the witness $w$ given a pair of different accepting conversations on the same input $x$ such that $(x,w)\in R$. Thus, no prover without knowledge of the witness can convince the verifier. Third, for a proof of knowledge, we need the \textit{special honest-verifier zero-knowledge} property, meaning that there exists a ppt algorithm that on input $x$ outputs an accepting conversation with the same probability distribution as the conversations between $P$ and a honest $V$ on input $x$. Thus, the verifier is not able to distinguish the witness, if it acts honestly. For the extended Schnorr proof in Figure \ref{extschnorr}, completeness, special soundness and special honest-verifier zero-knowledge hold as can be shown by similar arguments as in \cite{80}.\\
By the interactive procedure from Figure \ref{rangetest}, the prover shows to the verifier, that the commited value $x$ lies in the interval $[0,2^l-1]$ without revealing anything else about $x$. For the security of the construction in Figure \ref{extschnorr} we refer to \cite{80}, where a more general protocol was considered (particularly the special honest verifier zero-knowledge property is needed). We use the Fiat-Shamir heuristic \cite{83} to make the Range test non-interactive.

\begin{figure}[hpbt]
\fbox{
\begin{minipage}{13.15cm}
\begin{center}\scriptsize
\underline{Extended Schnorr Proof of Knowledge}
\end{center}\scriptsize
Let $G$ be a cyclic group of order $q$ and $h$ a generator of $G$. The prover wants to show to the verifier that she knows the discrete logarithms of either $R=h^r$ or $S=h^s$ without revealing its value and position. W.l.o.g. the prover knows $r$.
\begin{description}\scriptsize
\item \textbf{\mbox{\upshape \sffamily Com}}: The prover chooses random $v_2,w_2,z_1\leftarrow_R\mathbb{Z}_q^*$ and sends $a_1=h^{z_1}, a_2=h^{w_2}\cdot S^{-v_2}$ to the verifier.
\item \textbf{\mbox{\upshape \sffamily Chg}}: The verifier sends a random $v\leftarrow_R\mathbb{Z}_q^*$ to the prover.
\item \textbf{\mbox{\upshape \sffamily Opn}}: The prover sends $v_1=v-v_2, w_1=z_1+v_1 r, v_2$ and $w_2$ to the verifier.
\item \textbf{\mbox{\upshape \sffamily Chk}}: The verifier verifies that $v=v_1+v_2$ and that $h^{w_1}=a_1 R^{v_1}, h^{w_2}=a_2 S^{v_2}$.
\end{description}
\end{minipage}}
\caption{Schnorr Proof of Knowledge of one out of two secrets.}\label{extschnorr}
\end{figure} 

\begin{figure}[hpbt]
\fbox{
\begin{minipage}{13.15cm}
\begin{center}\scriptsize
\underline{Range test}
\end{center}
\begin{description}\scriptsize
\item \textbf{\mbox{\upshape \sffamily Gen}}: $(\mbox{\upshape\sffamily pp},pk)\leftarrow \mbox{\upshape\sffamily Gen}(1^\kappa)$, with public parameters $\mbox{\upshape\sffamily pp}$ describing the cyclic group $G$ of order $q$, a test range $[0,2^l-1]\subset G$ and public key $pk=(g,h)$ for generators $g,h$ of $G$ (the prover does not know $dlog_g(h)$).
\item \textbf{\mbox{\upshape \sffamily Com}}: For $x=\sum_{k=0}^{l-1} x^{(k)}\cdot 2^k$ with $x^{(k)}\in\{0,1\}$ for all $k=0,\ldots,l-1$, the prover chooses random $r^{(k)},v_2^{(k)},w_2^{(k)},z_1^{(k)}\leftarrow_R\mathbb{Z}_q^*, k=0\ldots,l-1$, computes $r\equiv\sum_{k=0}^{l-1} r^{(k)}\cdot 2^k\text{ mod } q$ and sends
\begin{align*} 
com=g^x\cdot h^r\text{ mod } q,\, com^{(k)}= & g^{x^{(k)}}\cdot h^{r^{(k)}}\text{ mod } q\\
a_1^{(k)}= & h^{z_1^{(k)}}\\
a_2^{(k)}= & h^{w_2^{(k)}}\cdot(h^{r^{(k)}}g^{2x^{(k)}-1})^{-v_2^{(k)}}
\end{align*}
for $k=0,\ldots,l-1$ to the verifier.
\item \textbf{\mbox{\upshape \sffamily Chg}}: The verifier verifies that $com\equiv\prod_{k=0}^{l-1} (com^{(k)})^{2^k}$ mod $q$. The verifier sends random $v^{(k)}\leftarrow_R\mathbb{Z}_q^*$ for all $k=0\ldots,l-1$ to the prover.
\item \textbf{\mbox{\upshape \sffamily Opn}}: The prover sends $v_1^{(k)}=v^{(k)}-v_2^{(k)}, w_1^{(k)}=z_1^{(k)}+v_1^{(k)}r^{(k)}, v_2^{(k)}$ and $w_2^{(k)}$ for all $k=0\ldots,l-1$ to the verifier.
\item \textbf{\mbox{\upshape \sffamily Chk}}: For all $k=0\ldots,l-1$, the verifier verifies that $v^{(k)}=v_1^{(k)}+v_2^{(k)}$ and that either\\ $(h^{w_1^{(k)}},h^{w_2^{(k)}})= (a_1^{(k)}(com^{(k)})^{v_1^{(k)}},a_2^{(k)}(com^{(k)}\cdot g^{-1})^{v_2^{(k)}})$ or\\ $(h^{w_1^{(k)}},h^{w_2^{(k)}})=(a_1^{(k)}(com^{(k)}\cdot g^{-1})^{v_1^{(k)}},a_2^{(k)}(com^{(k)})^{v_2^{(k)}})$.
\end{description}
\end{minipage}}
\caption{Range test for a commited value.}\label{rangetest}
\end{figure}

\subsection{Secure Computation}

Our security definition for Investcoin will be twofold. In the first part, we consider a honest-but-curious coalition consisting of the untrusted system administration together with its compromised investors and the group of honest investors. Here we refer to notions from Secure Multi-Party Computation (SMPC). In the second part, we identify reasonable malicious behaviour that an investor could possibly execute and show, how the system can be secured against such malicious investors.\\
In this Section we focus on defining security against the honest-but-curious coalition.

\begin{Def}\label{hbcsecdef} Let $\kappa$ be a security parameter and $n, \lambda\in\mathbb{N}$ with $n=\text{poly}(\kappa)$. Let $\rho$ be a protocol executed by a group of size $u\leq n$ and a coalition of honest-but-curious adversaries of size $n-u+1$ for computing the deterministic functionality $f_\rho$. The protocol $\rho$ performs a \text{\upshape secure computation} (or \text{\upshape securely computes} $f_\rho$), if there exists a ppt algorithm $\mathcal{S}$, such that $\left\{\mathcal{S}(1^\kappa,y,f_\rho(x,y))\right\}_{x,y,\kappa}\approx_c\left\{\text{\upshape\sffamily view}^\rho(x,y,\kappa)\right\}_{x,y,\kappa}$,
where $\text{\upshape\sffamily view}^\rho(x,y,\kappa)=(y,r,m)$ is the view of the coalition during the execution of the protocol $\rho$ on input $(x,y)$, $x$ is the input of the group, $y$ is the input of the coalition, $r$ is its random tape and $m$ is its vector of received messages.
\end{Def}

This definition follows standard notions from SMPC (as in \cite{81}) and is adapted to our environment: first, we consider two-party protocols where each party consists of multiple individuals (each individual in a party has the same goals) and second, we do not consider security of the coalition against the group, since the system administration has no input and thus its security against honest-but-curious investors is trivial. Rather we will later consider its security against malicious investors.\\
Since Investcoin is the combination of various protocols, we will prove the security of these protocols separately and then use the composition theorem by Canetti \cite{82}.

\begin{Thm}[Composition Theorem in the Honest-but-Curious Model \cite{82}]\label{compthmhbcsec}
Let $\kappa$ be a security parameter and let $m=\text{poly}(\kappa)$. Let $\pi$ be a protocol that computes a functionality $f_\pi$ by making calls to a trusted party computing the functionalities $f_1,\ldots,f_m$. Let $\rho_1,\ldots,\rho_m$ be protocols computing the functionalities $f_1,\ldots,f_m$ respectively. Denote by $\pi^{\rho_1,\ldots,\rho_m}$ the protocol $\pi$, where the calls to a trusted party are replaced by executions of $\rho_1,\ldots,\rho_m$. If $\pi,\rho_1,\ldots,\rho_m$ non-adaptively perform secure computations, then also $\pi^{\rho_1,\ldots,\rho_m}$ non-adaptively performs a secure computation.
\end{Thm}

\section{Investcoin}\label{invcnsec}

In this section, we define Distributed Investment Encryption (DIE) and introduce Investcoin. This protocol is build from a combination of the PSA scheme from Figure \ref{DDHEXM1}, the homomorphic Commitment scheme from Figure \ref{pedersenexm} and the Range test from Figure \ref{rangetest}. Moreover, we provide a simple key-generation protocol for the Investcoin protocol that allows the dynamic join and leave of investors and is fault-tolerant towards investors.

\subsection{Definition of DIE}

A DIE scheme is a SMPC protocol that allows a group of $n$ investors to aggregate their investments for a set of $\lambda$ projects in a way that only the sums of the investments for every single project and the sums of the investments of every single investor are revealed. In particular, no sum of investments from a set of investors of size smaller than $n$ to any set of projects is revealed. Additionally, each investor may receive back a return for her particular investments in a way that only the sum of her returns are revealed.

\begin{Def}[Distributed Investment Encryption]
Let $\kappa$ be a security parameter and $n, \lambda\in\mathbb{N}$ with $n=\text{\upshape\sffamily poly}(\kappa)$ and $\lambda=\text{\upshape\sffamily poly}(\kappa)$. A \text{\upshape Distributed}\linebreak \text{\upshape Investment Encryption (DIE) scheme} $\Omega=(\mbox{\upshape\sffamily Setup}, \mbox{\upshape\sffamily DIEEnc}, \mbox{\upshape\sffamily DIECom}, \mbox{\upshape\sffamily DIETes},$\linebreak $\mbox{\upshape\sffamily DIEUnvPay},\mbox{\upshape\sffamily DIEDec}, \mbox{\upshape\sffamily DIEUnvRet})$ is defined by the following ppt algorithms:
\begin{description}
\item \textbf{\mbox{\upshape \sffamily Setup}}: $(\mbox{\upshape\sffamily pp},T,sk_0,sk_1,\ldots,sk_n,pk_1,\ldots,pk_n)\leftarrow \mbox{\upshape\sffamily Setup}(1^\kappa)$ with public parameters $\mbox{\upshape\sffamily pp}$, $T=\{t_1,\ldots,t_\lambda\}$, secret keys $sk_0,sk_1\ldots,sk_n$ and public keys $pk_1,\ldots,pk_n$.
\item \textbf{\mbox{\upshape \sffamily DIEEnc}}: For all $j=1,\ldots,\lambda$ and $i=1,\ldots,n$, choose $x_{i,j}\in\mathcal{D}$ (from some finite set $\mathcal{D}$ parameterised by $\mbox{\upshape\sffamily pp}$) and compute: 
\[c_{i,j}\leftarrow \mbox{\upshape\sffamily DIEEnc}_{sk_i}(t_j,x_{i,j}).\]
\item \textbf{\mbox{\upshape \sffamily DIECom}}: For all $j=1,\ldots,\lambda$ and $i=1,\ldots,n$, choose $r_{i,j}\leftarrow\mathcal{U}(R)$ (i.e. a value chosen uniformly at random from some finite set $R$ parameterised by $\mbox{\upshape\sffamily pp}$) and compute
\[(com_{i,j},\tilde{c}_{i,j})\leftarrow\mbox{\upshape\sffamily DIECom}_{pk_i,sk_i}(x_{i,j}, r_{i,j}).\]
\item \textbf{\mbox{\upshape \sffamily DIETes}}: For all $j=1,\ldots,\lambda$ and $i=1,\ldots,n$: for a public key $pk_i$ and data value $x_{i,j}$, compute 
\[b_{T,i,j}\leftarrow\mbox{\upshape \sffamily DIETes}_{pk_i}(x_{i,j}),\, b_{T,i,j}\in\{0,1\}.\]
\item \textbf{\mbox{\upshape \sffamily DIEUnvPay}}: For all $i=1,\ldots,n$: for a public commitment key $pk_i$, commitment values $com_{i,1},\ldots,com_{i,\lambda}$, a data value $C_i$ and an opening value $D_i$, compute
\[b_{P,i}\leftarrow\mbox{\upshape\sffamily DIEUnvPay}_{pk_i}\left(\prod_{j=1}^\lambda com_{i,j},C_i,D_i\right),\, b_{P,i}\in\{0,1\}.\]
\item \textbf{\mbox{\upshape \sffamily DIEDec}}: For all $j=1,\ldots,\lambda$, ciphertexts $c_{1,j},\ldots,c_{n,j}$, compute
\[X^\prime_j=\mbox{\upshape\sffamily DIEDec}_{sk_0}(t_j,c_{1,j},\ldots,c_{n,j}).\]
\item \textbf{\mbox{\upshape \sffamily DIEUnvRet}}: Generate a public return factor $\alpha_j$ for every $j\in\{1,\ldots,\lambda\}$. For all $i=1,\ldots,n$: for a public commitment key $pk_i$, commitment values $com_{i,1},\ldots,com_{i,\lambda}$, a data value $E_i$ and an opening value $F_i$, compute
\[b_{R,i}\leftarrow\mbox{\upshape\sffamily DIEUnvRet}_{pk_i}\left(\prod_{j=1}^\lambda com_{i,j}^{\alpha_j},E_i,F_i\right),\, b_{R,i}\in\{0,1\}.\]
\end{description}
\end{Def}

The system parameters $\mbox{\upshape\sffamily pp}$ are public and constant for all $t_j\in T$. Every investor $N_i$ encrypts her investment amount $x_{i,j}$ for project $P_j$ with the secret key $sk_i$ in the $\mbox{\upshape \sffamily DIEEnc}$ algorithm and sends the ciphertext to the system administrator. If the system administrator receives the ciphertexts from \textit{all} investors for some $t_j$, it can compute the sum of the investment amounts for $P_j$ using the decryption key $sk_0$ in the decryption algorithm $\mbox{\upshape \sffamily DIEDec}$. After having computed the overall investment amounts for every project, the system administrator needs to know how much each investor has invested in total in order to collect this total amount from the investor's bank account. Each investor $N_i$ claims to the system administrator her total amount $C_i$ invested in all the projects $P_1,\ldots,P_\lambda$ together. In order to prove that $C_i=\sum_{j=1}^\lambda x_{i,j}$, $N_i$ sends a DIE commitment to the system administrator with the $\mbox{\upshape \sffamily DIECom}$ algorithm, where the commitment value $com_{i,j}$ is generated using the public key and the commitment cipher $\tilde{c}_{i,j}$ is generated using the secret key. The system administrator verifies that a combination of all commitments of each investor is valid for $C_i$ with the payment verification algorithm $\mbox{\upshape \sffamily DIEUnvPay}$. Furthermore, $\mbox{\upshape \sffamily DIETes}$ is executed in order to prove that the amounts are larger than or equal $0$. At maturity, for each $j=1,\ldots,\lambda$, each investor $N_i$ should receive back a multiple $\alpha_j x_{i,j}$ of her amount invested in project $P_j$ (e.g. a ROI). The factor $\alpha_j$ is publicly released by the management of project $P_j$ and may denote a rate of return, interest or similar. Each investor $N_i$ claims to the system administrator her total amount $E_i$ to receive back from all projects $P_1,\ldots,P_\lambda$ together. In order to prove that $E_i=\sum_{j=1}^\lambda \alpha_j x_{i,j}$, for all $i=1,\ldots,n$, the system administrator verifies that a combination of all commitments provided previously during the commitment phase (see $\mbox{\upshape \sffamily DIECom}$ algorithm) by each investor together with the corresponding return factors $\alpha_1,\ldots,\alpha_\lambda$ is valid for $E_i$ with the return verification algorithm $\mbox{\upshape \sffamily DIEUnvRet}$.

\subsection{Construction of Investcoin} 

\begin{figure}[hpbt]
\fbox{
\begin{minipage}{13.9cm}
\begin{center}\scriptsize
\underline{Investcoin}
\end{center}\tiny
Let $\kappa$ be a security parameter and $n, \lambda\in\mathbb{N}$ with $n=\text{\upshape\sffamily poly}(\kappa)$ and $\lambda=\text{\upshape\sffamily poly}(\kappa)$. Let $\Sigma=(\text{\mbox{\upshape \sffamily Setup}},\mbox{\upshape \sffamily PSAEnc},\mbox{\upshape \sffamily PSADec})$ be the PSA scheme from Figure $\ref{DDHEXM1}$ and let $\Gamma=(\mbox{\upshape \sffamily GenCom},\mbox{\upshape \sffamily Com},\mbox{\upshape \sffamily Unv})$ be the Commitment scheme from Figure $\ref{pedersenexm}$. We define Investcoin $\Omega=(\mbox{\upshape\sffamily DIESet}, \mbox{\upshape\sffamily DIEEnc}, \mbox{\upshape\sffamily DIECom},\mbox{\upshape\sffamily DIETes}, \mbox{\upshape\sffamily DIEUnvPay},\mbox{\upshape\sffamily DIEDec}, \mbox{\upshape\sffamily DIEUnvRet})$ as follows.
\begin{description}\tiny
\item \textbf{\mbox{\upshape \sffamily DIESet}}: 
\begin{itemize}\tiny
\item The system administrator generates public parameters $\mbox{\upshape\sffamily pp}=\{q,p,H\}$ with primes $q>m\cdot n, p=2q+1$ (where $m=2^l-1$ is the maximum possible amount to invest into a single project by an investor, $l=\text{\upshape\sffamily poly}(\kappa)$) and parameters $T=\{t_0,t_1,\ldots,t_\lambda\}$, $\tilde{T}=\{\tilde{t}_1,\ldots,\tilde{t}_\lambda\}$.
\item The system administrator generates a pair of generators $(h_1,h_2)\in\mathbb{Z}_{p^2}^*\times\mathbb{Z}_{p^2}^*$ with $\text{ord}(h_1)=\text{ord}(h_2)=pq$, sets the public key $pk=(\tilde{T},h_1,h_2)$  (which is the same for all investors) and defines a hash function $H:T\cup\tilde{T}\to\mathcal{QR}_{p^2}$.
\item The system administrator and the investors together generate secret keys $s_0$ and\linebreak $sk_i=(s_i,\tilde{s}_i)\leftarrow_R\mathbb{Z}_{pq}\times\mathbb{Z}_{pq}$ for all $i=1,\ldots,n$ with $s_0\equiv-\sum_{i=1}^n s_i\Mod pq$.
\item The system administrator generates secret parameters $\beta_0,\ldots,\beta_\lambda\leftarrow_R [-q^\prime,q^\prime], q^\prime<q/(m\lambda)$, such that $\prod_{j=0}^\lambda H(t_j)^{\beta_j}\equiv\prod_{j=1}^\lambda H(\tilde{t}_j)^{\beta_j}\equiv 1\Mod p^2$ (see Section $\ref{invkeygensec}$ for the details). It sets\linebreak $sk_0=(s_0,\beta_0,\ldots,\beta_\lambda)$.
\end{itemize}
\item \textbf{\mbox{\upshape \sffamily DIEEnc}}: For all $j=1,\ldots,\lambda$, each investor $N_i$ chooses $x_{i,j}\in[0,m]$ and $x_{i,0}=0$ and for all $j=0,\ldots,\lambda$, $N_i$ sends the following ciphertexts to the system administrator: 
\[c_{i,j}\leftarrow \mbox{\upshape\sffamily DIEEnc}_{sk_i}(t_j,x_{i,j})=\mbox{\upshape\sffamily PSAEnc}_{s_i}(t_j,x_{i,j}).\]
\item \textbf{\mbox{\upshape \sffamily DIECom}}: For all $j=1,\ldots,\lambda$, each $N_i$ chooses $r_{i,j}\leftarrow_R [0,m]$ and sends the following to the system administrator:
\[(com_{i,j},\tilde{c}_{i,j})\leftarrow\mbox{\upshape\sffamily DIECom}_{pk,sk_i}(x_{i,j}, r_{i,j})=(\mbox{\upshape\sffamily Com}_{pk}(x_{i,j},r_{i,j}),\mbox{\upshape\sffamily PSAEnc}_{\tilde{s}_i}(\tilde{t}_j,r_{i,j})).\]
\item \textbf{\mbox{\upshape \sffamily DIETes}}: For all $j=1,\ldots,\lambda$, the algorithm from Figure $\ref{rangetest}$ on $(x_{i,j}, r_{i,j})$ in $\mathcal{QR}_{p^2}$ is executed between the system administrator (as verifier) and each investor $N_i$ (as prover) using $pk$ to compute $b_{T,i,j}\in\{0,1\}$, where $b_{T,i,j}=1$, if the proof is accepted and $b_{T,i,j}=0$, if not. Note that there always exists a random representation $r_{i,j}^{(0)},\ldots,r_{i,j}^{(l-1)}\leftarrow_R [0,m]$, such that $r_{i,j}\equiv\sum_{k=0}^{l-1} r_{i,j}^{(k)}\cdot 2^k\Mod pq$.
\item \textbf{\mbox{\upshape \sffamily DIEUnvPay}}: The system administrator verifies for each investor $N_i$ with commitment values $(com_{i,1},\tilde{c}_{i,1}),\ldots,(com_{i,\lambda},\tilde{c}_{i,\lambda})$ and ciphers $c_{i,1},\ldots,c_{i,\lambda}$ that
\[\mbox{\upshape\sffamily Unv}_{pk}\left(\prod_{j=1}^\lambda com_{i,j}^{\beta_j},A_i,B_i\right)=1,\]
where $A_i=\left(\left(\prod_{j=0}^\lambda c_{i,j}^{\beta_j}\mbox{ mod } p^2\right)-1\right)/p$ and $B_i=\left(\left(\prod_{j=1}^\lambda \tilde{c}_{i,j}^{\beta_j}\mbox{ mod } p^2\right)-1\right)/p$ (else it aborts). Then each investor $N_i$ sends $C_i=\sum_{j=0}^\lambda x_{i,j}$ and $D_i=\sum_{j=1}^\lambda r_{i,j}$ to the system administrator which computes
\[b_{P,i}=\mbox{\upshape\sffamily DIEUnvPay}_{pk}\left(\prod_{j=1}^\lambda com_{i,j},C_i,D_i\right)=\mbox{\upshape\sffamily Unv}_{pk}\left(\prod_{j=1}^\lambda com_{i,j},C_i,D_i\right), b_{P,i}\in\{0,1\}\]
and charges the bank account of $N_i$ with the amount $C_i$, if $b_{P,i}=1$.
\item \textbf{\mbox{\upshape \sffamily DIEDec}}: For all $j=0,\ldots,\lambda$ and ciphertexts $c_{1,j},\ldots,c_{n,j}$, the system administrator computes
\[X_j=\mbox{\upshape\sffamily DIEDec}_{sk_0}(t_j,c_{1,j},\ldots,c_{n,j})=\mbox{\upshape\sffamily PSADec}_{s_0}(t_j,c_{1,j},\ldots,c_{n,j})\]
and verifies that $X_0=0$ (else it aborts).
\item \textbf{\mbox{\upshape \sffamily DIEUnvRet}}: The management of each project $P_j$ publishes a return factor $\alpha_j\in [-q^\prime,q^\prime]$. The system administrator charges the bank account of the management of each project $P_j$ with the amount $\alpha_j X_j$. Each investor $N_i$ sends $E_i=\sum_{j=1}^\lambda \alpha_j x_{i,j}$ and $F_i=\sum_{j=1}^\lambda \alpha_j r_{i,j}$ to the system administrator. If the verification in the $\mbox{\upshape \sffamily DIEUnvPay}$ algorithm has output $1$, the system administrator computes for all $i=1,\ldots,n$:
\[b_{R,i}=\mbox{\upshape\sffamily DIEUnvRet}_{pk}\left(\prod_{j=1}^\lambda com_{i,j}^{\alpha_j},E_i,F_i\right)=\mbox{\upshape\sffamily Unv}_{pk}\left(\prod_{j=1}^\lambda com_{i,j}^{\alpha_j},E_i,F_i\right), b_{R,i}\in\{0,1\}\]
and transfers the amount $E_i$ to the bank account of investor $N_i$, if $b_{R,i}=1$.
\end{description}
\end{minipage}}
\caption{The Investcoin protocol.}\label{invcnprot}
\end{figure}

The $\mbox{\upshape \sffamily DIESet}$ algorithm in Figure \ref{invcnprot} executes the Setup algorithms of the underlying schemes. Additionally, $\mbox{\upshape \sffamily DIESet}$ generates a verification parameter $\beta_j$ for each project $P_j$ (and an additional $\beta_0$ - this will be used for the security against malicious investors) which is only known to the system administration. In Section \ref{invkeygensec} we provide a simple protocol for generating the secrets. The encryption algorithm $\mbox{\upshape \sffamily DIEEnc}$ executes the encryption algorithm of $\Sigma$ and encrypts the amounts invested by $N_i$ into $P_j$. In order to prove that $C_i=\sum_{j=1}^\lambda x_{i,j}$, the $N_i$ execute the commitment algorithm $\mbox{\upshape \sffamily DIECom}$ commiting to the amounts $x_{i,j}$ invested using the randomness $r_{i,j}$ by executing the commitment algorithm of $\Gamma$ and encrypting the $r_{i,j}$ with $\Sigma$. The Range test algorithm $\mbox{\upshape \sffamily DIETes}$ ensures that the investments are larger or equal $0$. The payment verification algorithm $\mbox{\upshape\sffamily DIEUnvPay}$ first verifies that the combination of the commited amounts \textit{in the correct order} is valid for the same combination of amounts encrypted \textit{in the correct order} by executing the verification algorithm of $\Gamma$. If the investor has not cheated, this verification will output $1$ by the homomorphy of $\Gamma$ and the fact that $\prod_{j=0}^\lambda H(t_j)^{\beta_j}=\prod_{j=1}^\lambda H(\tilde{t}_j)^{\beta_j}=1$. The $\mbox{\upshape\sffamily DIEUnvPay}$ algorithm verifies that the combination of commitments is valid for the aggregate $C_i$ of the investments of $N_i$. The decryption algorithm $\mbox{\upshape \sffamily DIEDec}$ then decrypts the aggregated amounts for every project by executing the decryption algorithm of $\Sigma$. After the projects are realised, each investor $N_i$ should receive back a multiple $\alpha_j x_{i,j}$ of her amount invested in each project $P_j$ (e.g. a ROI). The factor $\alpha_j$ is publicly released by the management of project $P_j$ and denotes a rate of return, interest or similar. This value is equal for every investor, since only the investor's stake should determine how much her profit from that project is. If the first check in the $\mbox{\upshape\sffamily DIEUnvPay}$ algorithm has output $1$, the return verification algorithm $\mbox{\upshape\sffamily DIEUnvRet}$ verifies that the combination of commitments and return factors is valid for the claimed aggregate $E_i$ of the returns to receive by $N_i$.\\
We emphasize the low communication effort needed after the $\mbox{\upshape \sffamily DIESet}$ algorithm: every investor sends the messages for $\mbox{\upshape \sffamily DIEEnc}$, $\mbox{\upshape \sffamily DIECom}$, $\mbox{\upshape \sffamily DIETes}$ and $\mbox{\upshape \sffamily DIEUnvPay}$ in one shot to the system, later only the messages for $\mbox{\upshape \sffamily DIEUnvRet}$ have to be sent. Thus, there are only two rounds of communication between the investors and the system. 

\begin{Thm}[Correctness of Investcoin] Let $\Omega$ be the protocol in Figure $\ref{invcnprot}$. Then the following properties hold.
\begin{enumerate}
\item For all $j=1,\ldots,\lambda$ and $x_{1,j},\ldots,x_{n,j}\in[0,m]$:
\[\mbox{\upshape\sffamily DIEDec}_{sk_0}(t_j, \mbox{\upshape\sffamily DIEEnc}_{sk_1}(t_j,x_{1,j}),\ldots,\mbox{\upshape\sffamily DIEEnc}_{sk_n}(t_j,x_{n,j}))=\sum_{i=1}^n x_{i,j}.\]
\item For all $i=1,\ldots,n$ and $x_{i,1},\ldots,x_{i,\lambda}\in[0,m]$:
\begin{align*}
 & \mbox{\upshape\sffamily DIEUnvPay}_{pk}\left(\prod_{j=1}^\lambda com_{i,j},\sum_{j=1}^\lambda x_{i,j},\sum_{j=1}^\lambda r_{i,j}\right)=1\\
\Leftrightarrow & \exists\,(\tilde{c}_{i,1},\ldots,\tilde{c}_{i,\lambda})\,:\,(com_{i,j},\tilde{c}_{i,j})\leftarrow\mbox{\upshape\sffamily DIECom}_{pk,sk_i}(x_{i,j}, r_{i,j})\,\forall\,j=1,\ldots,\lambda.
\end{align*}
\item For all $i=1,\ldots,n$, public return factors $\alpha_1,\ldots,\alpha_\lambda$ and $x_{i,1},\ldots,x_{i,\lambda}\in[0,m]$:
\begin{align*}
 & \mbox{\upshape\sffamily DIEUnvRet}_{pk}\left(\prod_{j=1}^\lambda com_{i,j}^{\alpha_j},\sum_{j=1}^\lambda \alpha_j x_{i,j},\sum_{j=1}^\lambda \alpha_j r_{i,j}\right)=1\\
\Leftrightarrow & \exists\,(\tilde{c}_{i,1},\ldots,\tilde{c}_{i,\lambda})\,:\,(com_{i,j},\tilde{c}_{i,j})\leftarrow\mbox{\upshape\sffamily DIECom}_{pk,sk_i}(x_{i,j}, r_{i,j})\,\forall\,j=1,\ldots,\lambda.
\end{align*}
\end{enumerate}
\end{Thm}
\begin{proof}
The first correctness property is given by the correctness of the PSA scheme from Figure \ref{DDHEXM1}. The second and third correctness properties are given by the correctness and the homomorphy of the Commitment scheme from Figure \ref{pedersenexm}.
\hfill$\qed$\end{proof}

By the first property, the decryption of all ciphers results in the sum of the amounts they encrypt. So the projects receive the correct investments. By the second property, the total investment amount of each investor is accepted by the system if the investor has commited to it. Thus, the investor's account will be charged with the correct amount. By the third property, the total return to each investor is accepted by the system if the investor has commited to the corresponding investment amount before. Thus, the investor will receive the correct return on investment (ROI).

\subsection{Generation of public parameters and secret key generation protocol}\label{invkeygensec}

In this section, we show how the system sets the random oracle $H:T\to\mathcal{QR}_{p^2}$ and we provide a decentralised key generation protocol for Investcoin. It supports dynamic join, dynamic leave and fault-tolarance of investors using one round of communication between the investors. In the Section \ref{secproofsec}, we will show how to use the public parameters and the secret key generation protocol for the security of Investcoin.

\subsubsection{Setting the random oracle.}

Recall that we need to generate public parameters, a random oracle $H:T\to\mathcal{QR}_{p^2}$ and secret parameters $\beta_0,\ldots,\beta_\lambda\leftarrow_R [-q^\prime,q^\prime]$, $q^\prime<q/(m\lambda)$, such that for $t_0,\ldots,t_\lambda,\tilde{t}_1,\ldots,\tilde{t}_\lambda\in T$ the following equation holds.
\begin{equation}\label{newverificationparameters}
\prod_{j=0}^\lambda H(t_j)^{\beta_j}=\prod_{j=1}^\lambda H(\tilde{t}_j)^{\beta_j}=1.
\end{equation}
First, for $j=0,\ldots,\lambda-2$, on input $t_j$ let $H(t_j)$ be random in $\mathcal{QR}_{p^2}$ and for $j=1,\ldots,\lambda-2$, on input $\tilde{t}_j$ let $H(\tilde{t}_j)$ be random in $\mathcal{QR}_{p^2}$. The system chooses $\beta_0,\ldots,\beta_{\lambda-2}\leftarrow_R [-q^\prime,q^\prime], \beta_{\lambda-1}\leftarrow_R [-q^\prime,q^\prime-1]$ as part of its secret key (note that choosing these values according to a different distribution gives no advantage to the system). Then it computes
\begin{align*} (H(t_{\lambda-1}),H(\tilde{t}_{\lambda-1})) & =\left(\prod_{j=0}^{\lambda-2} H(t_j)^{\beta_j},\prod_{j=1}^{\lambda-2} H(\tilde{t}_j)^{\beta_j}\right),\\
(H(t_\lambda),H(\tilde{t}_\lambda),\beta_\lambda) & =(H(t_{\lambda-1}),H(\tilde{t}_{\lambda-1}),-1-\beta_{\lambda-1}),
\end{align*}
instructs each investor $N_i$ to set $x_{i,\lambda-1}=x_{i,\lambda}=0$ and sets $\alpha_{\lambda-1}=\alpha_\lambda=1$. In this way Equation \eqref{newverificationparameters} is satisfied. The projects $P_{\lambda-1}, P_\lambda$ deteriorate to 'dummy-projects' (e.g. if any investor decides to set $x_{i,\lambda}> 0$, then the system simply collects $X_\lambda>0$).

\subsubsection{Key generation for PSA within Investcoin.}

The building block for a key generation protocol is a $n-1$ out of $n$ secret sharing scheme between the investors and the system. It is executed before the first investment round in the following way.\\
For all $i=1,\ldots,n$, investor $N_i$ generates uniformly random values $s_{i,1},\ldots,s_{i,n}$ from the key space and sends $s_{i,i^\prime}$ to $N_{i^\prime}$ for all $i^\prime=1,\ldots,n$ via secure channel. Accordingly, each investor $N_{i^\prime}$ obtains the shares $s_{1,i^\prime},\ldots,s_{n,i^\prime}$. Then each investor $N_i$ sets the own secret key $s_i=\sum_{i^\prime=1}^n s_{i,i^\prime}$ and each investor $N_{i^\prime}$ sends $\sum_{i=1}^n s_{i,i^\prime}$ to the system. The system then computes
\[s_0=-\sum_{i^\prime=1}^n\left(\sum_{i=1}^n s_{i,i^\prime}\right)=-\sum_{i=1}^n\left(\sum_{i^\prime=1}^n s_{i,i^\prime}\right)=-\sum_{i=1}^n s_i.\]
By the secret sharing property this is a secure key generation protocol in the sense that only $N_i$ knows $s_i$ for all $i=1,\ldots,n$ and only the system knows $s_0$.\\
For key generation, each investor has to send one message to every other investor and one message to the system which makes $n^2$ messages for the total network.\\
As a drawback, note that the key of each single investor is controlled by the other investors together with the system: for example, if $N_1,\ldots,N_{n-1}$ (maliciously) send the shares $s_{1,n},\ldots,s_{n-1,n}$ and $s_{n,1},\ldots,s_{n,n-1}$ to the system, it can compute the entire key $s_n$ of $N_n$.\\
Assume that before the start of an arbitrary investment round, new investors want to join the network or some investors want to leave the network or some investors fail to send the required ciphers. In order to be able to carry out the protocol execution, the network can make a key update that requires $O(n)$ messages (rather than $O(n^2)$ messages for a new key setup) using the established secret sharing scheme. Due to space limitations, we omit the (simple) details.

\section{Security of Investcoin}

\subsection{Definition of security}

The administration is honest-but-curious and may compromise investors to build a coalition for learning the values of uncompromised investors. Investors who are not in the coalition may try to execute the following (reasonable) malicious behaviour:
\begin{enumerate}
\item Use different values for $x_{i,j}$ in $\mbox{\upshape\sffamily DIEEnc}$ and $\mbox{\upshape\sffamily DIECom}$ in order to have a larger profit than allowed.
\item Invest negative amounts $x_{i,j}<0$ in order to 'steal' funds from the projects.
\item Use different parameters than generated in the Setup-phase (i.e. send inconsistent or simply random messages) in order to distort the whole computation.
\end{enumerate}

\begin{Def}[Privacy-Preserving Investment system]\label{ppisecurity}
Let $\kappa$ be a security parameter and $n, \lambda\in\mathbb{N}$ with $n=\text{\upshape\sffamily poly}(\kappa)$ and $\lambda=\text{\upshape\sffamily poly}(\kappa)$. A DIE scheme $\Omega=(\mbox{\upshape\sffamily DIESet}, \mbox{\upshape\sffamily DIEEnc},$ $\mbox{\upshape\sffamily DIECom}, \mbox{\upshape\sffamily DIETes}, \mbox{\upshape\sffamily DIEUnvPay}, \mbox{\upshape\sffamily DIEDec}, \mbox{\upshape\sffamily DIEUnvRet})$ consisting of ppt algorithms and executed between a group of uncorrupted parties of size $u\leq n$ and a coalition of honest-but-curious adversaries of size $n-u+1$ is a \text{\upshape Privacy-Preserving Investment} (\mbox{\upshape\sffamily PPI}) system, if the following properties hold.
\begin{enumerate}
\item Let $f_\Omega$ be the deterministic functionality computed by $\Omega$. Then $\Omega$ performs a secure computation (according to Definition $\ref{hbcsecdef}$).
\item $\Omega$ provides \text{\upshape linkage}, i.e. for all $i=1,\ldots,n$, $sk_i, pk_i, C_i, D_i, E_i, F_i,$\linebreak $T=\{t_1,\ldots,t_\lambda\}$, $(x_{i,j})_{j=1,\ldots,\lambda}$, $(com_{i,j})_{j=1,\ldots,\lambda}$, $(\alpha_j)_{j=1,\ldots,\lambda}$ the following holds: if
\begin{align*}
& c_{i,j}\leftarrow \mbox{\upshape\sffamily DIEEnc}_{sk_i}(t_j,x_{i,j})\,\forall\, j=1,\ldots,\lambda\\
\wedge & \mbox{\upshape\sffamily DIEUnvPay}_{pk_i}\left(\prod_{j=1}^\lambda com_{i,j},C_i,D_i\right)=1\\
\wedge & \mbox{\upshape\sffamily DIEUnvRet}_{pk_i}\left(\prod_{j=1}^\lambda com_{i,j}^{\alpha_j},E_i,F_i\right)=1,
\end{align*}
then
\[C_i=\sum_{j=1}^\lambda x_{i,j}\wedge E_i=\sum_{j=1}^\lambda \alpha_j x_{i,j}.\]
\item For all $i=1,\ldots,n$ and $j=1,\ldots,\lambda$:  $\mbox{\upshape \sffamily DIETes}_{pk_i}(x_{i,j})=1$ iff $x_{i,j}\geq 0$.
\item For all $i=1,\ldots,n$, there is a ppt distinguisher $\mathcal{D}_i$, such that the following holds. For any $x_{i,j},r_{i,j}\in [0,m]$, let $c_{i,j}\leftarrow\mbox{\upshape\sffamily DIEEnc}_{sk_i}(t_j,x_{i,j})$\linebreak and $(com_{i,j},\tilde{c}_{i,j})\leftarrow\mbox{\upshape\sffamily DIECom}_{pk_i,sk_i}(x_{i,j}, r_{i,j})$ for all $j=1,\ldots,\lambda$.\linebreak For any $x^\prime_{i,j},r^\prime_{i,j}\in [0,m]$, let $c^\prime_{i,j}\leftarrow \mbox{\upshape\sffamily DIEEnc}_{sk^{(j)}_i}(t^\prime_j,x^\prime_{i,j})$ and\linebreak $(com^\prime_{i,j},\tilde{c}^\prime_{i,j})\leftarrow\mbox{\upshape\sffamily DIECom}_{pk^{(j)}_i,sk^{(j)}_i}(x^\prime_{i,j}, r^\prime_{i,j})$, for all $j=1,\ldots,\lambda$, such that\linebreak $(pk^{(j)}_i,sk^{(j)}_i,t^\prime_j)\neq(pk_i,sk_i,t_j)$ for at least one $j\in[\lambda]$, i.e. at least one entry of at least one tuple is different. Then
\begin{align*} & \left|\Pr\left[\mathcal{D}_i\left(1^\kappa,sk_0,c_{i,1},com_{i,1},\tilde{c}_{i,1},\ldots,c_{i,\lambda},com_{i,\lambda},\tilde{c}_{i,\lambda}\right)=1\right]\right.\\
- & \left.\Pr\left[\mathcal{D}_i\left(1^\kappa,sk_0,c^\prime_{i,1},com^\prime_{i,1},\tilde{c}^\prime_{i,1},\ldots,c^\prime_{i,\lambda},com^\prime_{i,\lambda},\tilde{c}^\prime_{i,\lambda}\right)=1\right]\right|\\
\geq & 1-\mbox{\upshape\sffamily neg}(\kappa).
\end{align*}
The probability space is defined over the internal randomness of $\mathcal{D}_i$.
\end{enumerate}
\end{Def}

The definition is twofold: on the one hand it covers the security of honest investors against a honest-but-curious coalition consisting of the untrusted system administration and compromised investors (Property $1$) and on the other hand it covers the security of the system against maliciously behaving investors (Properties $2,3,4$). Note that we have to distinguish between these two requirements, since we assume different behaviours for the two groups of participants, i.e. we cannot simply give a real-world-ideal-world security proof as in the SMPC literature in the malicious model. Instead, we follow the notions of the SMPC literature \cite{81} for the security of honest investors (only) in the honest-but-curious model and additionally provide security notions against some reasonably assumable behaviour of malicious investors. For the security against a honest-but-curious coalition, the first property ensures that from the outcome of the decryption no other information than $X_j$ can be detected for all $j=1,\ldots,\lambda$ and that the single amounts comitted to by the investors for payment and return are hidden. For the security of the system against maliciously behaving investors, imagine the situation where an investor $N_i$ claims to having payed amount $x_{i,j}$ to project $P_j$ (in the $\mbox{\upshape\sffamily DIECom}$ algorithm) but in fact has only payed $\tilde{x}_{i,j}<x_{i,j}$ (in the $\mbox{\upshape\sffamily DIEEnc}$ algorithm). If the return factor $\alpha_j$ is larger than $1$, then $N_i$ would unjustly profit more from her investment than she actually should and the system would have a financial damage.\footnote{Usually the investor cannot know if $\alpha_j>1$ at the time of cheating, since it becomes public in the future. However, in the scenario where a priori information about $\alpha_j$ is known to some investors or where investors simply act maliciously, we need to protect the system from beeing cheated.} Therefore the second property says that for all $i=1,\ldots,n$ with overwhelming probability, whenever $x_{i,1},\ldots,x_{i,\lambda}$ were send by $N_i$ using the $\mbox{\upshape\sffamily DIEEnc}$ algorithm and $\mbox{\upshape\sffamily DIEUnvPay}$, $\mbox{\upshape\sffamily DIEUnvRet}$ accept $C_i, E_i$ respectively, then $C_i$ and $E_i$ must be the legitimate amounts that $N_i$ respectively has invested in total and has to get back as return in total. The third property ensures that no investor is able to perform a negative investment. The fourth property ensures that all investors use the correct parameters as generated by the \mbox{\upshape \sffamily DIESet} algorithm.\footnote{More precisely, it ensures that a cheating investor will be identified by the system.}

\subsection{Proof of security}\label{secproofsec}

First, we concentrate on the security against the honest-but-curious coalition (Property $1$) and then show security against malicious investors (Propertis $2,3,4$).

\begin{Thm} By the DDH assumption in the group $\mathcal{QR}_{p^2}$ of quadratic residues modulo $p^2$ for a safe prime $p$, Investcoin is a \text{\upshape\sffamily PPI} system in the random oracle model.
\end{Thm}
\begin{proof} The proof follows from Lemma $\ref{h-b-c}$, Lemma $\ref{linkage}$, Lemma $\ref{positiveness}$ and Lemma $\ref{correctparams}$ below.
\hfill$\qed$\end{proof}

\subsubsection{Security against honest-but-curious coalition.}

Investcoin is the combination of the protocols described in Figures $\ref{DDHEXM1}, \ref{pedersenexm}, \ref{extschnorr}$.\footnote{The Range test in Figure \ref{rangetest} is the combination of the protocols in Figure \ref{pedersenexm} and \ref{extschnorr}. Therefore it suffices to show the security of these protocols.} We first show the security of the underlying protocols seperately and then use Theorem \ref{compthmhbcsec} in order to show composition. In this section, assume without loss of generality that the indices $i=1,\ldots,u$ belong to the group of uncorrupted investors and the indices $i=u+1,\ldots,n$ belong to the investors corrupted by the system administrator. Before showing security, we briefly explain the functionalities used in the proofs. In the following, computations are performed modulo $p^2$.
\begin{itemize}
\item The functionality $f_\Sigma$ is computed by the protocol in Figure $\ref{DDHEXM1}$. It takes as input $n$ values $x_1,\ldots,x_n$ and outputs their sum \[f_\Sigma(x_1,\ldots,x_n)=\sum_{i=1}^n x_i.\]
\item The functionality $f_\Gamma$ is computed by the protocol in Figure $\ref{pedersenexm}$. It takes as input three values $com,x,r$ and outputs $b=1$ iff $(x,r)$ is a valid opening for the commitment $com$, else $b=0$.
\item The functionality $f_\Upsilon$ is computed by the protocol in Figure $\ref{extschnorr}$. It takes as input the values $r,R$ and outputs $b=1$ iff $r$ is a discrete logarithm of either $R$ or $S=R\cdot g^{-1}$ with base $h$, else $b=0$.
\item For $\boldsymbol\beta=(\beta_0,\ldots,\beta_\lambda)$, the functionality $f_{\boldsymbol\beta}$ is computed by the $\text{\upshape\sffamily DIEUnvPay}$ algorithm in Figure \ref{invcnprot} (using $\beta_0,\ldots,\beta_\lambda$ and $(c_{i,j}, \tilde{c}_{i,j})$ for all $i=1,\ldots,n$, $j=1,\ldots,\lambda$ generated by $\text{\upshape\sffamily DIEEnc}, \text{\upshape\sffamily DIECom}$ respectively). It takes as input the values $x_0$, $x_1, r_1,\ldots,x_\lambda, r_\lambda$ and outputs $A=\sum_{j=0}^\lambda \beta_j x_j$ and $B=\sum_{j=1}^\lambda \beta_j r_j$.
\end{itemize}

\begin{Lem}\label{PSAseclem} Under the DDH assumption, the protocol $\Sigma$ in Figure~$\ref{DDHEXM1}$ securely computes the functionality $f_\Sigma$ in the random oracle model.
\end{Lem}
\begin{proof} Assume the coalition is corrupted. Let $\{x_{i,j}\,:\,i=1,\ldots,n, j=1,\ldots,\lambda\}$ be the input messages of all investors. We construct a simulator $\mathcal{S}$ for the 
\[\text{\upshape\sffamily view}^\Sigma((x_{i,j})_{i=1,\ldots,n,j=1,\ldots,\lambda},\kappa)=((x_{i,j})_{i=u+1,\ldots,n,j=1,\ldots,\lambda},\perp,(c_{i,j})_{i=1,\ldots,u,j=1,\ldots,\lambda})\]
of the coaltion, where $c_{i,j}=\text{\upshape\sffamily PSAEnc}_{s_i}(t_j,x_{i,j})$ for all $i,j$ and\\ $\text{\upshape\sffamily PSAEnc}_{s_0}(t_j,0)\cdot\prod_i c_{i,j}=1+p\cdot f_\Sigma((x_{i,j})_{i=1,\ldots,n})$ mod $p^2$ for all $j$, as follows.\footnote{There is no internal randomness of the coalition. Instead the common public randomness from the random oracle $H$ is used.}
\begin{description}
\item $\mathcal{S}$ takes as input $1^\kappa$, %secret keys $\{sk_i\,:\,i=u+1,\ldots,n\}$, 
messages $(x_{i,j})_{i=u+1,\ldots,n, j=1,\ldots,\lambda}$ and the protocol outputs $(X_j)_{j=1,\ldots,\lambda}=(f_\Sigma((x_{i,j})_{i=1,\ldots,n}))_{j=1,\ldots,\lambda}$.
\begin{enumerate}
\item Choose messages $(x^\prime_{i,j})_{i=1,\ldots,u, j=1,\ldots,\lambda}$ with each $x^\prime_{i,j}\in[0,2^l-1]$, such that for all $j=1,\ldots,\lambda$: \[\sum_{i=1}^u x_{i,j}^\prime = X_j-\sum_{i=u+1}^n x_{i,j}.\]
\item Compute $c^\prime_{i,j}=\text{\upshape\sffamily PSAEnc}_{s_i}(t_j,x_{i,j}^\prime)$ for $i=1,\ldots,u$ and\linebreak $c_{i,j}=\text{\upshape\sffamily PSAEnc}_{s_i}(t_j,x_{i,j})$ for $i=u+1,\ldots,n$, $j=1,\ldots,\lambda$ for the keys\linebreak $s_1,\ldots,s_n\leftarrow_R\mathbb{Z}_{pq}$ with $s_0\equiv -\sum_{i=1}^n s_i\Mod pq$.
\end{enumerate}
\item $\mathcal{S}$ outputs $(x_{i,j})_{i=u+1,\ldots,n, j=1,\ldots,\lambda},\perp,(c^\prime_{i,j})_{i=1,\ldots,u, j=1,\ldots,\lambda}$.
\end{description}
The computational indistinguishability of the output distribution of $\mathcal{S}$ from\linebreak $\left\{\text{\upshape\sffamily view}^\Sigma((x_{i,j})_{i=1,\ldots,n,j=1,\ldots,\lambda},\kappa)\right\}_{(x_{i,j})_{i=1,\ldots,n,j=1,\ldots,\lambda},\kappa}$ immediately follows from the $\text{\upshape\sffamily AO}1$ security of $\Sigma$, which holds by the DDH assumption.
\hfill$\qed$\end{proof}

\begin{Lem}\label{Tesseclem} The protocol $\Upsilon$ in Figure $\ref{extschnorr}$ securely computes the accepting functionality\footnote{I.e. the case $b=1$. In the case $b=0$, by the completeness property, the prover has tried to convince the verifier about a wrong statement. Thus, we can assume that the prover was intending some malicious behaviour. In this case, we do not care about preserving the privacy of this prover.} $f_\Upsilon$.
\end{Lem}
\begin{proof} Assume the honest verifier is corrupted and the prover knows the discrete logarithm of either $R$ or $S=R\cdot g^{-1}$ base $h$, where $g,h,$ are public. We construct a simulator $\mathcal{S}$ for the
\[\text{\upshape\sffamily view}^\Upsilon(t,R,S,\kappa)=(R,S,v,a_1,a_2,(v_1,w_1),(v_2,w_2))\]
of the verifier, such that $v=v_1+v_2$ and $h^{w_1}=a_1 R^{v_1}, h^{w_2}=a_2 S^{v_2}$ and $t$ is the input of the prover, as follows.
\begin{description}
\item $\mathcal{S}$ takes as input $1^\kappa$, %secret keys $\{sk_i\,:\,i=u+1,\ldots,n\}$, 
the verifier's input $(R,S)$ and the protocol output $b=1$.
\begin{enumerate}
\item Choose random values $v_1^\prime,v_2^\prime,w_1^\prime,w_2^\prime\leftarrow_R\mathbb{Z}_q$.
\item Set $v^\prime=v_1^\prime+v_2^\prime$.
\item Set $a_1^\prime=h^{w_1^\prime}R^{-v_1^\prime}$ and $a_2^\prime=h^{w_2^\prime}S^{-v_2^\prime}$.
\end{enumerate}
\item $\mathcal{S}$ outputs $R,S,v^\prime,a_1^\prime,a_2^\prime,(v_1^\prime,w_1^\prime),(v_2^\prime,w_2^\prime)$.
\end{description}
The value $v^\prime$ is random, since $v_1^\prime,v_2^\prime$ are random. Then the distribution of the output of $\mathcal{S}$ is indistinguishable from $\left\{\text{\upshape\sffamily view}^\Upsilon(t,R,S,\kappa)\right\}_{t,R,S,\kappa}$ by the special honest-verifier zero-knowledge property of $\Upsilon$, which holds perfectly.
\hfill$\qed$\end{proof}

\begin{Lem}\label{betacomp} Under the DDH assumption, the algorithms $\text{\upshape\sffamily DIEEnc}$, $\text{\upshape\sffamily DIECom}$ and\linebreak $\text{\upshape\sffamily DIEUnvPay}$ securely compute the functionality $f_{\boldsymbol\beta}$ in the random oracle model.
\end{Lem}
\begin{proof} In order to compute $A_i=\sum_{j=0}^\lambda \beta_j x_{i,j}$ and $B_i=\sum_{j=1}^\lambda \beta_j r_{i,j}$, the algorithm $\text{\upshape\sffamily DIEUnvPay}$ (executed by the system administrator) takes as input the values $\beta_0,\ldots,\beta_\lambda$ and the encryptions $c_{i,j}, \tilde{c}_{i,j}$ of $x_{i,j}, r_{i,j}$ respectively generated by the algorithms $\text{\upshape\sffamily DIEEnc}$ and $\text{\upshape\sffamily DIECom}$ (executed by investor $N_i$) for $j=(0),1,\ldots,\lambda$. Since the encryptions $c_{i,j}, \tilde{c}_{i,j}$ are both generated by the protocol $\Sigma$ in Figure~\ref{DDHEXM1}, the statement follows from the $\text{\upshape\sffamily AO}1$ security of $\Sigma$ as in the proof of Lemma~\ref{PSAseclem}.
\hfill$\qed$\end{proof}

\begin{Lem}\label{h-b-c} Under the DDH assumption, Investcoin satisfies Property~$1$ of Definition $\ref{ppisecurity}$ in the random oracle model.
\end{Lem}
\begin{proof} Let $\Omega^{f_\Sigma,f_\Upsilon,f_{\boldsymbol\beta}}$ be the protocol $\Omega$, where executions of the aforementioned protocols are replaced by calls to a trusted party computing the according functionalities $f_\Sigma,f_\Upsilon,f_{\boldsymbol\beta}$.\footnote{The system administrator, which is part of the coalition, and the investors simply feed the trusted party with their according inputs, outputs and received messages to compute the respective functionalities. Since we deal with modular sequential composition, the trusted party cannot be invoked to compute the functionalities simultaneously. Rather a functionality is executed after the output of the previous functionality was distributed to the parties.} Assume the coalition is corrupted. We first describe $\Omega^{f_\Sigma,f_\Upsilon,f_{\boldsymbol\beta}}$. Thereby we assume that the key generation in $\text{\upshape\sffamily DIESet}$ is completed as a pre-computation. The other algorithms work as follows.
\begin{itemize}
\item For all $j=0,\ldots,\lambda$, the input $(x_{i,j})_{i=1,\ldots,u}$ of the uncorrupted group of investors and the input $(x_{i,j})_{i=u+1,\ldots,n}$ of the coalition is sent to the trusted party which returns $X_j=\sum_{i=1}^n x_{i,j}=f_{\Sigma}(x_{1,j},\ldots,x_{n,j})$ to the system administrator. The system administrator verifies that $X_0=0$, otherwise it aborts.
\item For all $i=1,\ldots,u$, $j=1,\ldots,\lambda$, investor $N_i$ chooses $r_{i,j}\leftarrow_R [0,m]$ and sends $com_{i,j}=\mbox{\upshape\sffamily Com}_{pk}(x_{i,j},r_{i,j})$ to the system administrator (where \mbox{\upshape\sffamily Com} is as in Figure \ref{pedersenexm}).
\item For all $i=1,\ldots,u$, $j=1,\ldots,\lambda$, $k=0,\ldots,l-1$, investor $N_i$ chooses\linebreak $x^{(k)}_{i,j}\in\{0,1\}, r^{(k)}_{i,j}\leftarrow_R [0,m]$ with $x_{i,j}\equiv\sum_{k=0}^{l-1} x^{(k)}_{i,j}\cdot 2^k\Mod pq$ and\linebreak $r_{i,j}\equiv\sum_{k=0}^{l-1} r^{(k)}_{i,j}\cdot 2^k\Mod pq$ and sends $com^{(k)}_{i,j}=\mbox{\upshape\sffamily Com}_{pk}(x^{(k)}_{i,j},r^{(k)}_{i,j})$ to the system administrator. The system administrator verifies that\linebreak $com_{i,j}\equiv\prod_{k=0}^{l-1} (com^{(k)}_{i,j})^{2^k}\Mod p^2$ for all $i,j$ and sends the $com^{(k)}_{i,j}$ to the trusted party. Moreover, the investors send the $x^{(k)}_{i,j}, r^{(k)}_{i,j}$ to the trusted party which returns \[f_\Upsilon(r^{(k)}_{i,j}, com^{(k)}_{i,j})=b_{T,i,j,k}\] for all $i=1,\ldots,u$, $j=1,\ldots,\lambda$, $k=0,\ldots,l-1$ to the system administrator.
\item For $i=1,\ldots,u$, investor $N_i$ sends $x_{i,j}, r_{i,j}$, $j=1,\ldots,\lambda$ to the trusted party and the system administrator sends $\boldsymbol\beta=(\beta_0,\ldots,\beta_\lambda)$ to the trusted party which returns 
\[(A_i,B_i)=\left(\sum_{j=1}^\lambda\beta_j x_{i,j},\sum_{j=1}^\lambda \beta_j r_{i,j}\right)=f_{\boldsymbol\beta}(0,x_{i,1},r_{i,1},\ldots,x_{i,\lambda},r_{i,\lambda})\] 
for all $i=1,\ldots,u$ to the system administrator which verifies that 
\[\mbox{\upshape\sffamily Unv}_{pk}\left(\prod_{j=1}^\lambda com_{i,j}^{\beta_j},A_i,B_i\right)=1.\]
Then the uncorrupted investors $N_i$, $i=1,\ldots,u$, send $C_i=\sum_{j=0}^\lambda x_{i,j}$ and $D_i=\sum_{j=1}^\lambda r_{i,j}$ to the system administrator which computes
\[b_{P,i}=\mbox{\upshape\sffamily Unv}_{pk}\left(\prod_{j=1}^\lambda com_{i,j},C_i,D_i\right), b_{P,i}\in\{0,1\}.\]
\item The uncorrupted investors $N_i$, $i=1,\ldots,u$, send $E_i=\sum_{j=0}^\lambda \alpha_j x_{i,j}$ and $F_i=\sum_{j=1}^\lambda \alpha_j r_{i,j}$ to the system administrator which computes
\[b_{R,i}=\mbox{\upshape\sffamily Unv}_{pk}\left(\prod_{j=1}^\lambda com^{\alpha_j}_{i,j},E_i,F_i\right), b_{R,i}\in\{0,1\}.\]
\end{itemize}

This is exactly the protocol in Figure \ref{invcnprot} where sub-protocols are exchanged by calls to a trusted party computing the ideal functionalities $f_\Sigma, f_\Upsilon, f_{\boldsymbol\beta}$ and the system administrator knows all secret values of the corrupted investors. We show that $\Omega^{f_\Sigma,f_\Upsilon,f_{\boldsymbol\beta}}$ performs a secure computation. In the following, we construct a simulator $\mathcal{S}$ for the $\text{\upshape\sffamily view}^{\Omega^{f_\Sigma,f_\Upsilon,f_{\boldsymbol\beta}}}((x_{i,j})_{i=1,\ldots,n,j=0,\ldots,\lambda},\kappa)=(y,r,m)$ of the coaltion with
\begin{align*}
 y= & ((x_{i,j})_{i=u+1,\ldots,n,j=0,\ldots,\lambda},(x^{(k)}_{i,j})_{i=u+1,\ldots,n,j=1,\ldots,\lambda,k=0,\ldots,l-1}),\\
 r= & ((r_{i,j})_{i=u+1,\ldots,n,j=1,\ldots,\lambda},(r^{(k)}_{i,j})_{i=u+1,\ldots,n,j=1,\ldots,\lambda,k=0,\ldots,l-1}),\\
 m= & ((com_{i,j})_{i=1,\ldots,u,j=1,\ldots,\lambda},(com_{i,j}^{(k)})_{i=1,\ldots,u,j=1,\ldots,\lambda,k=0,\ldots,l-1},(C_i,D_i,E_i,F_i)_{i=1,\ldots,u}).
\end{align*}

\noindent $\mathcal{S}$ takes as input $1^\kappa$, %secret keys $\{sk_i\,:\,i=u+1,\ldots,n\}$, 
 $y$ and the protocol outputs $X_j$ for $j=0,\ldots,\lambda$,\linebreak $\{b_{T,i,j,k}=1\,:\,i=1,\ldots,u,j=1,\ldots,\lambda,k=0,\ldots,l-1\}$, $A_i,B_i$ for $i=1,\ldots,u$.
\begin{enumerate}
\item For $i=1,\ldots,u$, $j=1,\ldots,\lambda$, set $x^\prime_{i,0}=0$ and choose $x^\prime_{i,j}\in[0,2^l-1]$, such that $\sum_{i=1}^u x_{i,j}^\prime = X_j-\sum_{i=u+1}^n x_{i,j}$ for all $j=1,\ldots,\lambda$, $\sum_{j=1}^\lambda \beta_j x_{i,j}^\prime=A_i$ for all $i=1,\ldots,u$. Choose $(x^{\prime^{(0)}}_{i,j},\ldots,x^{\prime^{(l-1)}}_{i,j})$ with $x^\prime_{i,j}\equiv\sum_{k=0}^{l-1}x^{\prime^{(k)}}_{i,j}\cdot 2^k\Mod pq$ for all $i=1,\ldots,u, j=1,\ldots,\lambda$.
\item For $i=1,\ldots,n, j=1,\ldots,\lambda$, choose $r^\prime_{i,j}\leftarrow_R[0,2^l-1]$ and $(r^{\prime^{(0)}}_{i,j},\ldots,r^{\prime^{(l-1)}}_{i,j})$ with $r^\prime_{i,j}\equiv\sum_{k=0}^{l-1}r^{\prime^{(k)}}_{i,j}\cdot 2^k\Mod pq$ and $\sum_{j=1}^\lambda \beta_j r_{i,j}^\prime=B_i$ for all $i=1,\ldots,n$, $j=1,\ldots,\lambda$.
\item Compute $com^\prime_{i,j}=\text{\upshape\sffamily Com}_{pk}(x_{i,j}^\prime,r^\prime_{i,j})$ and $com^{\prime^{(k)}}_{i,j}=\text{\upshape\sffamily Com}_{pk}(x_{i,j}^{\prime^{(k)}},r^{\prime^{(k)}}_{i,j})$ for all $i=1,\ldots,u$, $j=1,\ldots,\lambda$, $k=0,\ldots,l-1$.
\item For $i=1,\ldots,u$, compute $C_i^\prime=\sum_{j=1}^\lambda x_{i,j}^\prime$, $D_i^\prime=\sum_{j=1}^\lambda r_{i,j}^\prime$, $E_i^\prime=\sum_{j=1}^\lambda \alpha_j x_{i,j}^\prime$, $F_i^\prime=\sum_{j=1}^\lambda \alpha_j r_{i,j}^\prime$.
\end{enumerate}
$\mathcal{S}$ outputs $(y^\prime,r^\prime,m^\prime)$ with
\begin{align*}
 y^\prime= & ((x^\prime_{i,j})_{i=u+1,\ldots,n,j=0,\ldots,\lambda},(x^{\prime^{(k)}}_{i,j})_{i=u+1,\ldots,n,j=1,\ldots,\lambda,k=0,\ldots,l-1}),\\
 r^\prime= & ((r^\prime_{i,j})_{i=u+1,\ldots,n,j=1,\ldots,\lambda},(r^{\prime^{(k)}}_{i,j})_{i=u+1,\ldots,n,j=1,\ldots,\lambda,k=0,\ldots,l-1}),\\
 m^\prime= & ((com^\prime_{i,j})_{i=1,\ldots,u,j=1,\ldots,\lambda},(com_{i,j}^{\prime^{(k)}})_{i=1,\ldots,u,j=1,\ldots,\lambda,k=0,\ldots,l-1},(C^\prime_i,D^\prime_i,E^\prime_i,F^\prime_i)_{i=1,\ldots,u})
\end{align*}
and the output of $\mathcal{S}$ is indistinguishable from the view of the coalition by the perfect hiding property of $\Gamma$. The statement of the lemma follows from Lemma~$\ref{PSAseclem}$, Lemma $\ref{Tesseclem}$, Lemma $\ref{betacomp}$,  and Theorem $\ref{compthmhbcsec}$.
\hfill$\qed$\end{proof}

\subsubsection{Security against malicious investors.}

\begin{Lem}\label{linkage} Under the dlog assumption, Investcoin satisfies Property~$2$ of Definition $\ref{ppisecurity}$ with overwhelming probability in the random oracle model.
\end{Lem}
\begin{proof}
For any $i\in[n]$, assume that $c_{i,j}\leftarrow \mbox{\upshape\sffamily DIEEnc}_{sk_i}(t_j,x_{i,j})$ for all $j=1,\ldots,\lambda$,
$\mbox{\upshape\sffamily DIEUnvPay}_{pk}\left(\prod_{j=1}^\lambda com_{i,j},C_i,D_i\right)=1$,
$\mbox{\upshape\sffamily DIEUnvRet}_{pk}\left(\prod_{j=1}^\lambda com_{i,j}^{\alpha_j},E_i,F_i\right)=1$. Then by the rules of the protocol also $\mbox{\upshape\sffamily Unv}_{pk}\left(\prod_{j=1}^\lambda com_{i,j}^{\beta_j},A_i,B_i\right)=1$ for some $\beta_0,\ldots,\beta_\lambda\in[-q^\prime,q^\prime]$ (unknown to $N_i$) with $q^\prime<q/(m\lambda)$, where we have defined
\[A_i=\left(\left(\prod_{j=0}^\lambda c_{i,j}^{\beta_j}\Mod p^2\right)-1\right)/p\mbox{ and } B_i=\left(\left(\prod_{j=1}^\lambda \tilde{c}_{i,j}^{\beta_j}\Mod p^2\right)-1\right)/p\] 
with $c_{i,j}\equiv(1+p\cdot x_{i,j})\cdot H(t_j)^{s_i}\Mod p^2$ and $\tilde{c}_{i,j}\equiv(1+p\cdot r_{i,j})\cdot H(\tilde{t}_j)^{\tilde{s}_i}\Mod p^2$ for all $j=1,\ldots,\lambda$ and for a random oracle $H:T\cup\tilde{T}\to\mathcal{QR}_{p^2}$. Now assume that either $C_i\neq\sum_{j=1}^\lambda x_{i,j}$ or $E_i\neq\sum_{j=1}^\lambda\alpha_j x_{i,j}$. Then by the homomorphy and computational binding of the Pedersen commitment scheme $\Gamma$, which holds by the dlog assumption, there is at least one $j^\prime\in[\lambda]$, such that $com_{i,j^\prime}$ is not the commitment of $x_{i,j^\prime}$ but of some $x^\prime_{i,j^\prime}\neq x_{i,j^\prime}$. We define the mapping 
\[f_{\beta_0,\ldots,\beta_\lambda} : [0,m]^\lambda\to [-q,q], f_{\beta_0,\ldots,\beta_\lambda}(x_0,\ldots,x_\lambda)=\sum_{j=0}^\lambda \beta_j x_j.\]
Then, by Equation \eqref{newverificationparameters}, $f_{\beta_0,\ldots,\beta_\lambda}(x_{i,0},\ldots,x_{i,\lambda})=A_i$. Moreover, $f_{\beta_0,\ldots,\beta_\lambda}$ is not efficiently computable by $N_i$, since it does not know $\beta_0,\ldots,\beta_\lambda\in [-q^\prime,q^\prime]$, if $q^\prime$ is super-polynomial in $\kappa$ (from Equation \eqref{newverificationparameters}, $N_i$ can compute the $\beta_j$ only with negligible probability, since $H$ is a random oracle). Therefore with overwhelming probability, for any choice of $(x^\prime_{i,0},\ldots,x^\prime_{i,\lambda})$ by $N_i$, such that $com_{i,j}$ is valid for $x^\prime_{i,j}$ for all $j=1,\ldots,\lambda$, the equation $f_{\beta_0,\ldots,\beta_\lambda}(x^\prime_{i,0},\ldots,x^\prime_{i,\lambda})=A_i$ is only true, if $x^\prime_{i,j}=x_{i,j}$ for all $j=0,\ldots,\lambda$. This is a contradiction to the existence of at least one $j^\prime$ with $x^\prime_{i,j^\prime}\neq x_{i,j^\prime}$. By the binding of $\Gamma$, the case $f_{\beta_0,\ldots,\beta_\lambda}(x^\prime_{i,0},\ldots,x^\prime_{i,\lambda})\neq A_i$ is a contradiction to $\mbox{\upshape\sffamily Unv}_{pk}\left(\prod_{j=1}^\lambda com_{i,j}^{\beta_j},A_i,B_i\right)=1$.
\hfill$\qed$\end{proof}

\begin{Lem}\label{positiveness} Under the dlog assumption, Investcoin satisfies Property~$3$ of Definition $\ref{ppisecurity}$ with overwhelming probability.
\end{Lem}
\begin{proof} We have to show that for all $i=1,\ldots,n, j=1,\ldots,\lambda$ and $pk=(h_1,h_2)$, with overwhelming probabilty, $\mbox{\upshape \sffamily DIETes}_{pk}(x_{i,j})=1$ iff $x_{i,j}\geq 0$. The Range test in Figure \ref{rangetest} performs $l$ times the protocol from Figure \ref{extschnorr}. I.e. for every bit $x_{i,j}^{(k)}$ of $x_{i,j}$, the Range test performs an Extended Schnorr protocol (Figure \ref{extschnorr}) for proving the knowledge of $1$ out of $2$ secrets: either the prover knows the opening $(0,r_{i,j}^{(k)})$ to the Pedersen commitment $com_{i,j}^{(k)}=h_2^{r_{i,j}^{(k)}}$ or the opening $(1,r_{i,j}^{(k)})$ to $com_{i,j}^{(k)}=h_1\cdot h_2^{r_{i,j}^{(k)}}$, where $k=0,\ldots,l-1$.\\ 
Since the Pedersen Commitment is homomorphic, the commitment $com_{i,j}$ for $x_{i,j}$ is congruent $\prod_{k=0}^{l-1} (com_{i,j}^{(k)})^{2^k}$. Then the direction \textquotedblleft$\Leftarrow$\textquotedblright$\mbox{\,}$ immediately follows from the completeness property of the protocol in Figure \ref{extschnorr}.\\
Assume a malicious prover tries to prove that some commited value\linebreak $x_{i,j}\notin[0,2^l-1]$ is in $[0,2^l-1]$. Then there exists at least one $k^\prime\geq l$ such that $x_{i,j}=x_{i,j}^{(k^\prime)}\cdot 2^{k^\prime}+\sum_{k=0}^{l-1} x_{i,j}^{(k)}\cdot 2^{k}$ with $x_{i,j}^{(k^\prime)}, x_{i,j}^{(k)}\in\{0,1\}$ for all $k=0,\ldots,l-1$ (since we reduce modulo $pq$). Since the verifier has to receive commitments to exactly $l$ bits, the prover needs to know the discrete logarithm of at least one $h_2^{r_{i,j}^{(k^{\prime\prime})}}$ or $h_1\cdot h_2^{r_{i,j}^{(k^{\prime\prime})}}$ which is not among $com_{i,j}^{(k)}$, $k=0,\ldots,l-1,k^\prime$. The prover can compute it only with negligible probability by the dlog assumption. Thus, the direction \textquotedblleft$\Rightarrow$\textquotedblright$\mbox{\,}$ follows from the special soundness of the protocol from Figure~\ref{extschnorr}.
\hfill$\qed$\end{proof}

\begin{Lem}\label{correctparams} Under the dlog assumption, Investcoin satisfies Property~$4$ of Definition $\ref{ppisecurity}$ with overwhelming probability in the random oracle model.
\end{Lem}
\begin{proof}
Let $T=\{t_1,\ldots,t_\lambda\}$ be the set of public parameters used for encryption and let $pk$ be the public key for the commitment. For all $i\in[n]$, such that investor $N_i$ is not in the coalition with the system administrator, i.e. $i\in U$, let $sk_i$ be the secret key of investor $N_i$. We construct $\mathcal{D}_i$ as follows. On input \[1^\kappa,sk_0,c_{i,0},c^*_{i,1},com^*_{i,1},\tilde{c}^*_{i,1},\ldots,c^*_{i,\lambda},com^*_{i,\lambda},\tilde{c}^*_{i,\lambda},\] 
it has to decide whether $(c^*_{i,j},com^*_{i,j},\tilde{c}^*_{i,j})=(c_{i,j},com_{i,j},\tilde{c}_{i,j})$ for all $j=1,\ldots,\lambda$ or $(c^*_{i,j},com^*_{i,j},\tilde{c}^*_{i,j})=(c^\prime_{i,j},com^\prime_{i,j},\tilde{c}^\prime_{i,j})$ for all $j=1\ldots,\lambda$, where\linebreak $c_{i,j}\leftarrow\mbox{\upshape\sffamily PSAEnc}_{s_i}(t_j,x_{i,j})$, $com_{i,j}\leftarrow\mbox{\upshape\sffamily Com}_{pk}(x_{i,j}, r_{i,j})$, $\tilde{c}_{i,j}\leftarrow\mbox{\upshape\sffamily PSAEnc}_{\tilde{s}_i}(\tilde{t}_j,r_{i,j})$, $c^\prime_{i,j}\leftarrow \mbox{\upshape\sffamily PSAEnc}_{s^{(j)}_i}(t^\prime_j,x^\prime_{i,j})$, $com^\prime_{i,j}\leftarrow\mbox{\upshape\sffamily Com}_{pk^{(j)}}(x^\prime_{i,j}, r^\prime_{i,j})$, $\tilde{c}^\prime_{i,j}\leftarrow\mbox{\upshape\sffamily PSAEnc}_{\tilde{s}^{(j)}_i}(\tilde{t}^\prime_j,r^\prime_{i,j})$ for all $j=1,\ldots,\lambda$, s.t. $(pk^{(j)},s^{(j)}_i,\tilde{s}^{(j)}_i,t^\prime_j,\tilde{t}^\prime_j)\neq(pk,s_i,\tilde{s}_i,t_j,\tilde{t}_j)$ for at least one $j\in[\lambda]$, i.e. there is a difference in at least one entry for at least one tuple. $\mathcal{D}_i$ uses its input to compute 
\begin{align*} A_i & =\left(\left(c_{i,0}^{\beta_0}\cdot\prod_{j=1}^\lambda c_{i,j}^{*^{\beta_j}}\mbox{ mod } p^2\right)-1\right)/p\mbox{ and }\\
 B_i & =\left(\left(\prod_{j=1}^\lambda \tilde{c}_{i,j}^{*^{\beta_j}}\mbox{ mod } p^2\right)-1\right)/p.
\end{align*}
Then it computes and outputs $b=\mbox{\upshape\sffamily Unv}_{pk}\left(\prod_{j=1}^\lambda com_{i,j}^{*^{\beta_j}},A_i,B_i\right)$.
For the first case, i.e. if $(c^*_{i,j},com^*_{i,j},\tilde{c}^*_{i,j})=(c_{i,j},com_{i,j},\tilde{c}_{i,j})$ for all $j=1,\ldots,\lambda$, we have $b=1$ by construction. We have to show that in the second case, i.e. if $(c^*_{i,j},com^*_{i,j},\tilde{c}^*_{i,j})=(c^\prime_{i,j},com^\prime_{i,j},\tilde{c}^\prime_{i,j})$ for all $j=1\ldots,\lambda$, with overwhelming probability $b=0$. Thus, $\mathcal{D}_i$ will distinguish the cases on its input.\\
As a pre-computation, based on the key generation protocol from Section \ref{invkeygensec}, for all $i^\prime=1,\ldots,n$, investor $N_{i^\prime}$ has published \[T_{i,i^\prime}=\mbox{\upshape\sffamily PSAEnc}_{s_{i,i^\prime}}(t_0,0)\] on a black board for all the key shares $s_{1,i^\prime},\ldots,s_{n,i^\prime}$ received during the key generation and the system administrator has verified that \[\mbox{\upshape\sffamily PSADec}_{-\sum_{i=1}^n s_{i,i^\prime}}(t_0,T_{1,i^\prime},\ldots,T_{n,i^\prime})=0,\]  
i.e. all verifications are true (if not, the system aborts and the process is repeated with all but the excluded cheating investors) as described above. Thus,
\[c_{i,0}=\mbox{\upshape\sffamily PSAEnc}_{s_i}(t_0,0)=\prod_{i^\prime=1}^n T_{i,i^\prime}\] 
is fixed, i.e. $c_{i,0}$ is generated with the correct parameters from the first case. Now consider the second case, where $(c^*_{i,j},com^*_{i,j},\tilde{c}^*_{i,j})=(c^\prime_{i,j},com^\prime_{i,j},\tilde{c}^\prime_{i,j})$ for all $j=1\ldots,\lambda$. Recall that the system administrator's secret values $\beta_0,\ldots,\beta_\lambda$ are not efficiently computable by the investors (who do not collaborate with the system administrator). Assume $b=1$. Then 
\[h_1^{A_i} h_2^{B_i}=\prod_{j=1}^\lambda com_{i,j}^{*^{\beta_j}},\] 
where $(h_1,h_2)$ is the public key for the Pedersen commitment which is used for verification by the system administrator \textit{in both cases}. By the computational binding of the Pedersen commitment, which holds under the dlog assumption, $N_i$ must have known $A_i, B_i$ for computing a valid commitment $\prod_{j=1}^\lambda com_{i,j}^{*^{\beta_j}}$ (note that in the first case, this is not necessary, since all ciphers and commitments were generated with consistent parameters). %This is because $c_{i,0}$ is involved in the computation of $A_i$ and weighted with $\beta_0$. $\beta_0$ can be anything and cannot be controlled by $N_i$, except with probability at most $1/q^\prime$, which is negligible. Thus, the only chance of $N_i$ to eliminate the dependene of the commitment on $\beta_0$ is to satisfy Equation \eqref{newverificationparameters}. Therefore $t^\prime_j=t_j$ for all $j=1,\ldots,\lambda$, since $H$ is a random oracle and then $s^{(j)}_i=s_i$ for all $j=1,\ldots,\lambda$. Then $pk^{(j)}=pk$ for all $j=1,\ldots,\lambda$, since $N_i$ does not know $\text{\upshape dlog}_{h_1}(h_2)$ under the dlog assumption. Hence, all parameters are as in the first case, which is a contradiction.
This is not possible, except with negligible probability, since from $N_i$'s point of view, $\beta_0,\ldots,\beta_\lambda$ are unknown.
\hfill$\qed$\end{proof}

\addcontentsline{toc}{chapter}{Literatur}
\bibliography{Literatur}

\begin{thebibliography}{10}

\bibitem{67}
Emmanuel~A. Abbe, Amir~E. Khandani, and Andrew~W. Lo.
\newblock Privacy-preserving methods for sharing financial risk exposures.
\newblock {\em Am. Economic Review}, 102(3):65--70, 2012.

\bibitem{38seq}
Fabrice Benhamouda, Marc Joye, and Beno{\^i}t Libert.
\newblock A new framework for privacy-preserving aggregation of time-series
  data.
\newblock {\em ACM Transactions on Information and System Security}, 18(3),
  2016.

\bibitem{65}
Avrim Blum, Jamie Morgenstern, Ankit Sharma, and Adam Smith.
\newblock Privacy-preserving public information for sequential games.
\newblock In {\em Proc. of ITCS '15}, pages 173--180, 2015.

\bibitem{72}
Manuel Blum.
\newblock Coin flipping by telephone.
\newblock In {\em Proc. of CRYPTO '81}, pages 11--15, 1981.

\bibitem{56}
Fabrice Boudot.
\newblock Efficient proofs that a committed number lies in an interval.
\newblock In {\em Proc. of EUROCRYPT '00}, pages 431--444, 2000.

\bibitem{73}
Gilles Brassard, David Chaum, and Claude Cr{\'e}peau.
\newblock Minimum disclosure proofs of knowledge.
\newblock {\em J. Comput. Syst. Sci.}, 37(2):156--189, 1988.

\bibitem{57}
Jan Camenisch, Rafik Chaabouni, and Abhi Shelat.
\newblock Efficient protocols for set membership and range proofs.
\newblock In {\em Proc. of ASIACRYPT '08}, pages 234--252, 2008.

\bibitem{82}
Ran Canetti.
\newblock {Security and Composition of Multiparty Cryptographic Protocols}.
\newblock {\em Journal of Cryptology}, 13:143--202, 2000.

\bibitem{64}
Financial Crisis~Inquiry Comission.
\newblock {\em The Financial Crisis Inquiry Report: Final Report of the
  National Commission on the Causes of the Financial and Economic Crisis in the
  United States}, 2011.

\bibitem{80}
Ronald Cramer, Ivan Damg{\aa}rd, and Berry Schoenmakers.
\newblock Proofs of partial knowledge and simplified design of witness hiding
  protocols.
\newblock In {\em Proc. of CRYPTO '94}, pages 174--187, 1994.

\bibitem{83}
Amos Fiat and Adi Shamir.
\newblock How to prove yourself: Practical solutions to identification and
  signature problems.
\newblock In {\em Proc. of CRYPTO '86}, pages 186--194, 1987.

\bibitem{66}
Mark Flood, Jonathan Katz, Stephen Ong, and Adam Smith.
\newblock Cryptography and the economics of supervisory information: Balancing
  transparency and confidentiality.
\newblock {\em Federal Reserve Bank of Cleveland, Working Paper no. 13-11},
  2013.

\bibitem{81}
Oded Goldreich.
\newblock {\em Foundations of Cryptography: Volume 2, Basic Applications}.
\newblock Cambridge University Press, 2004.

\bibitem{63}
Nicola Jentzsch.
\newblock {\em The Economics and Regulation of Financial Privacy - A
  Comparative Analysis of the United States and Europe}.
\newblock 2001.

\bibitem{38}
Marc Joye and Beno{\^\i}t Libert.
\newblock A scalable scheme for privacy-preserving aggregation of time-series
  data.
\newblock In {\em Proc. of FC '13}, pages 111--125. 2013.

\bibitem{77}
Ian Miers, Christina Garman, Matthew Green, and Aviel~D. Rubin.
\newblock Zerocoin: Anonymous distributed e-cash from bitcoin.
\newblock In {\em Proc. of SP '13}, pages 397--411, 2013.

\bibitem{76}
Satoshi Nakamoto.
\newblock Bitcoin: A peer-to-peer electronic cash system.

\bibitem{62}
Michael Nofer.
\newblock {\em The Value of Social Media for Predicting Stock Returns -
  Preconditions, Instruments and Performance Analysis}.
\newblock PhD thesis, Techn. Univ. Darmstadt, 2014.

\bibitem{71}
Torben~P. Pedersen.
\newblock Non-interactive and information-theoretic secure verifiable secret
  sharing.
\newblock In {\em Proc. of CRYPTO '91}, pages 129--140, 1991.

\bibitem{58}
Kun Peng, Colin Boyd, Ed~Dawson, and Eiji Okamoto.
\newblock A novel range test.
\newblock pages 247--258, 2006.

\bibitem{59}
Kun Peng and Ed~Dawson.
\newblock A range test secure in the active adversary model.
\newblock In {\em Proc. of ACSW '07}, pages 159--162, 2007.

\bibitem{79}
Claus-Peter Schnorr.
\newblock Efficient identification and signatures for smart cards.
\newblock In {\em Proc. of CRYPTO '89}, pages 239--252, 1989.

\bibitem{2}
Elaine Shi, T.{-}H.~Hubert Chan, Eleanor~G. Rieffel, Richard Chow, and Dawn
  Song.
\newblock Privacy-preserving aggregation of time-series data.
\newblock In {\em Proc. of NDSS '11}, 2011.

\bibitem{68}
Filipp Valovich.
\newblock On the hardness of the learning with errors problem with a discrete
  reproducible error distribution.
\newblock {\em CoRR abs/1605.02051}, 2016.

\bibitem{40}
Filipp Valovich and Francesco Ald{\`{a}}.
\newblock Private stream aggregation revisited.
\newblock {\em CoRR abs/1507.08071}, 2015.

\end{thebibliography}

\appendix
\section{Preservation of Market liquidity}\label{markliqsec}

We show that it is still possible to privately exchange any part of an investment between any pair of investors within Investcoin, i.e. market liquidity is unaffected. Assume that an investment round is over but the returns are not yet executed, i.e. the system administrator already received $c_{i,j},(com_{i,j},\tilde{c}_{i,j}),C_i,D_i$ for all $i=1,\ldots,n$ and $j=1,\ldots,\lambda$, but not $E_i,F_i$. Assume further that for some $i,i^\prime\in\{1,\ldots,n\}$, investors $N_i$ and $N_{i^\prime}$ confidentially agree on a transfer of amount $x_{(i,i^\prime),j}$ (i.e. a part of $N_i$'s investment in project $P_j$) from investor $N_i$ to investor $N_{i^\prime}$. This fact needs to be confirmed by the protocol in order to guarantee the correct returns from project $P_j$ to investors $N_i$ and $N_{i^\prime}$. Therefore the commitments to the invested amounts $x_{i,j}$ and $x_{i^\prime,j}$ respectively need to be updated. For the update, $N_i$ and $N_{i^\prime}$ agree on a value $r_{(i,i^\prime),j}\leftarrow_R [0,m]$ via secure channel. This value should be known only to $N_i$ and $N_{i^\prime}$. Then $N_i$ and $N_{i^\prime}$ respectively compute
\begin{align*} com_{i,j}^\prime & \leftarrow\mbox{\upshape\sffamily Com}_{pk}(x_{(i,i^\prime),j}, r_{(i,i^\prime),j}),\\ com_{i^\prime,j}^\prime & \leftarrow\mbox{\upshape\sffamily Com}_{pk}(x_{(i,i^\prime),j}, r_{(i,i^\prime),j})
\end{align*}
and send their commitments to the system administrator which verifies that $com_{i,j}^\prime=com_{i^\prime,j}^\prime$.
%\[\mbox{\upshape\sffamily DIEUnvPay}_{pk}\left(com_{i,j}^\prime\cdot(com_{i^\prime,j}^\prime)^{-1},0,0\right)=1.\] 
Then the system administrator updates 
\[com_{i,j} \mbox{ by } com_{i,j}\cdot(com_{i,j}^\prime)^{-1},\]
which is possible since the $\mbox{\upshape\sffamily Com}$ algorithm is injective, and 
\[com_{i^\prime,j} \mbox{ by } com_{i^\prime,j}\cdot com_{i^\prime,j}^\prime.\] 
As desired, the updated values commit to $x_{i,j}-x_{(i,i^\prime),j}$ and to $x_{i^\prime,j}+x_{(i,i^\prime),j}$ respectively. Moreover, $N_i$ updates the return values $(E_i, F_i)$ by 
\[(E_i-\alpha_j\cdot x_{(i,i^\prime),j}, F_i-\alpha_j\cdot r_{(i,i^\prime),j})\] and $N_{i^\prime}$ updates $(E_{i^\prime}, F_{i^\prime})$ by
\[(E_{i^\prime}+\alpha_j\cdot x_{(i,i^\prime),j}, F_{i^\prime}+\alpha_j\cdot r_{(i,i^\prime),j}).\] 
The correctness of the update is guaranteed by Property $2$ and the confidentiality of the amount $x_{(i,i^\prime),j}$ (i.e. only $N_i$ and $N_{i^\prime}$ know $x_{(i,i^\prime),j}$) is guaranteed by Property $1$ of Definition \ref{ppisecurity} which are satisfied by Lemma \ref{linkage} and Lemma \ref{h-b-c} respectively.\\
Note that in general, this procedure allows a short sale for $N_i$ when $x_{(i,i^\prime),j}>x_{i,j}$ or for $N_{i^\prime}$ when $x_{(i,i^\prime),j}<0$ and $|x_{(i,i^\prime),j}|>x_{{i^\prime},j}$ (over the integers). If this behaviour is not desired, it may also be necessary to perform a Range test for the updated commitments $com_{i,j}\cdot(com_{i,j}^\prime)^{-1}$ (between the system administrator and $N_i$) and $com_{i^\prime,j}\cdot com_{i^\prime,j}^\prime$ (between the system administrator and $N_{i^\prime}$) to ensure that they still commit to amounts $\geq 0$.

\end{document}